\documentclass[11pt]{article}

\usepackage{epsfig}
\usepackage{fullpage}
\usepackage{amsmath}
\usepackage{amsfonts}
\usepackage{amssymb}
\usepackage{amsthm}
\usepackage{xspace}
\usepackage[usenames]{color}
\usepackage{srcltx}
\usepackage[colorlinks=true,linkcolor=blue,citecolor=red]{hyperref}

\def\draft{0}   

\ifnum\draft=0 
    \newcommand{\submitremove}[1]{}
\else
    \newcommand{\submitremove}[1]{#1}
\fi

\ifnum\draft=0 

\else

\fi

\ifnum\draft=1 
    \def\ShowAuthNotes{1}
\else
    \def\ShowAuthNotes{0}
\fi
\ifnum\ShowAuthNotes=1
\newcommand{\authnote}[2]{{ \footnotesize \bf{[#1's Note: #2]} }}
\else
\newcommand{\authnote}[2]{}
\fi

\newcommand{\cnote}[1]{{\authnote{Cynthia} {#1}}}
\newcommand{\gnote}[1]{{\authnote{Guy} {#1}}}

\newtheorem{theorem}{Theorem}[section]
\newtheorem{corollary}[theorem]{Corollary}
\newtheorem{remark}[theorem]{Remark}

\newtheorem{lemma}[theorem]{Lemma}
\newtheorem{proposition}{Proposition}[section]
\newtheorem{claim}[theorem]{Claim}
\newtheorem{fact}{Fact}[section]

\theoremstyle{definition}
\newtheorem{definition}{Definition}[section]

\newcommand{\Range}{{\rm Range}}

\newcommand{\polylog}{{\rm polylog}}

\newcommand{\E}{{\cal E}}

\newcommand{\M}{{\cal M}}
\newcommand{\N}{{\cal N}}

\newcommand{\X}{{\mathcal{X}}}

\newcommand{\Nt}{\mathbb{N}}

\newcommand{\sub}{\tau}
\newcommand{\eg}{{\it e.g.}}

\newcommand{\Cdp}{Concentrated differential privacy}
\newcommand{\cdp}{concentrated differential privacy}

\newcommand{\DKL}{D_{\mathit{KL}}}
\newcommand{\DsubG}{D_{\mathit{subG}}}
\newcommand{\Dmax}{D_{\infty}}
\newcommand{\Dapprox}{D_{\infty}^{\delta}}

\newcommand{\Var}{\mathit{Var}}
\DeclareMathOperator*{\Ex}{\mathbb{E}}

\newcommand{\by}{{ \bf y}}

\newcommand{\Rt}{\mathbb{R}}

\newcommand{\MD}{antipodal} 
\newcommand{\MDcaps}{Antipodal}

\newcommand{\DB}{{\it d}}

\newcommand{\dba}{x}
\newcommand{\dbb}{y}
\newcommand{\dbc}{z}

\newcommand{\eqdef}{\mathbin{\stackrel{\rm def}{=}}}
\newcommand{\Supp}{\mathrm{Supp}}

\newcommand{\sign}{{\mathit sign}}
\newcommand{\signa}{\sign(\alpha)}

\newcommand{\Loss}{L}

\newcommand{\ex}{\mathrm{E}} 

\newcommand{\eps}{\varepsilon}

\newcommand{\remove}[1]{}

\begin{document}
\title{Concentrated Differential Privacy}

\author{Cynthia Dwork \and Guy N. Rothblum}

\date{}

\maketitle

\begin{abstract}
We introduce {\em Concentrated Differential Privacy}, a relaxation of Differential Privacy enjoying better accuracy than both pure differential privacy and its popular ``$(\eps,\delta)$'' relaxation without compromising on cumulative privacy loss over multiple computations.
\end{abstract}

\section{Introduction}
\label{sec:intro}

The Fundamental Law of Information Recovery states, informally, that ``overly accurate'' estimates of ``too many'' statistics completely destroys privacy (\cite{DinurN03} {\it et sequelae}).  Differential privacy is a mathematically rigorous definition of privacy tailored to analysis of large datasets and equipped with a formal measure of privacy loss~\cite{DworkMNS06,Dwork06}. Moreover, differentially private algorithms take as input a parameter, typically called~$\eps$, that caps the permitted privacy loss in any execution of the algorithm and offers a concrete privacy/utility tradeoff. One of the strengths of differential privacy is the ability to reason about cumulative privacy loss over multiple analyses, given the values of $\eps$ used in each individual analysis.  By appropriate choice of $\eps$ it is possible to stay within the bounds of the Fundamental Law while releasing any given number of estimated statistics; however, before this work the bounds were not tight.

Roughly speaking, differential privacy ensures that the outcome of any anlysis on a database $\dba$ is distributed very similarly to the outcome on any neighboring
database $\dbb$ that differs from $\dba$ in just one row (Definition~\ref{def: DP}).  That is, differentially private algorithms are randomized, and in particular the {\em max divergence} between these two distributions (a sort maximum log odds ratio for any event; see Definition~\ref{def:Max-divergence} below) is bounded by the privacy parameter~$\eps$. This absolute guarantee on the maximum privacy loss is now sometimes referred to as ``pure'' differential privacy.
A popular relaxation, $(\eps,\delta)$-differential privacy~(Definition~\ref{def:delta-DP})\cite{DworkKMMN06}, guarantees that with probability at most $1-\delta$ the privacy loss does not exceed~$\eps$.\footnote{Note that the probability is over the coins of the algorithm performing the analysis. The above formulation is not immediate from the definition of $(\eps,\delta)$-differential privacy, this was proved in the full version of~\cite{DworkRV10}.}
Typically $\delta$ is taken to be ``cryptographically'' small, that is, smaller than the inverse of any polynomial in the size of the dataset, and pure differential privacy is simply the special case in which $\delta=0$.  The relaxation frequently permits
asymptotically better accuracy than pure differential privacy for the same value of~$\eps$, even when~$\delta$ is very small.

What happens in the case of multiple analyses?
While the composition of $k$ $(\eps,0)$-differentially privacy algorithms is at worst $(k\eps,0)$-differentially private, it is also simultaneously $(\sqrt{2k \ln(1/\delta)} \eps +k\eps(e^\eps - 1),\delta)$-differentially private {\em for every $\delta$}~\cite{DworkRV10}.  Let us deconstruct the statement of this result. First, privacy loss is a random variable that captures differences in the probability distributions obtained when a randomized algorithm $\M$ is run on $\dba$ as opposed to neighboring database $\dbb$ (see Section~\ref{sec:defs:DP}).  In general, if the max divergence between two distributions is bounded by $\eps$ then their KL-divergence (Definition~\ref{def:KL-divergence}) is bounded by~$\eps(e^\eps - 1)$.  This means that the expected privacy loss for a single $(\eps,0)$-differentially private computation is bounded by~$\eps(e^\eps - 1)$.  By linearity of expectation, the expected loss over $k$ $(\eps,0)$-differentially private algorithms is bounded by $k \eps(e^\eps-1)$.  The statement therefore says that the cumulative privacy loss random variable over $k$ computations is tightly concentrated about its mean: the probability of privacy loss exceeding its expectation by $t\sqrt{k} \eps$ falls exponentially in $t^2/2$ for all~$t \ge 0$.  We will return to this formulation presently.

More generally, we prove
the following Advanced Composition theorem, which improves on the composition theorem of Dwork, Rothblum, and Vadhan~\cite{DworkRV10} by exactly halving the bound on expected privacy loss of~$(\eps,0)$-differentially privacy mechanisms (the proof is otherwise identical). Details for the sharper bound appear in Section~\ref{sec: relationship to DP}.
\begin{theorem}
\label{thm: advanced composition}
For all $\eps, \delta, \delta' \ge 0$, the class of $(\eps, \delta')$-differentially private mechanisms satisfies
$( \sqrt{2k \ln (1/\delta)}\eps +  k\eps(e^\eps -1)/2, k\delta' + \delta)$-differential privacy under $k$-fold adaptive composition.
\end{theorem}
As the theorem shows (recall that $\delta'$ is usually taken to be ``sub-polynomially small''), the pure and relaxed forms of differential privacy behave quite similarly under composition.  For the all-important class of {\em counting} queries (``How many people in the dataset satisfy property~$P$?'') Theorem~\ref{thm: advanced composition} leads to accuracy bounds that differ from the bounds imposed by (one instantiation of) the Fundamental Law by a factor of roughly $\sqrt{2\log(1/\delta)}$.\footnote{For example, for $O(n)$ counting queries errors on the order of $o(\sqrt{n})$ on all but a $0.239$ fraction of queries leads to blatant non-privacy~\cite{DworkMT07}, while noise drawn from a Laplace distribution with standard deviation $O(\sqrt{n \ln(1/\delta)}/\eps)$ yields $(\eps, \delta)$-differential privacy~\cite{DworkMNS06,DworkRV10}. \gnote{What is the constant in the little-oh?}} Recently, tight bounds on the composition of $(\eps, \delta)$-differentially private algorithms have been given in \cite{KairouzOV15,MurtaghV16}(see below).


\paragraph{A New Relaxation.}
In this work we introduce a different relaxation, {\em Concentrated Differential Privacy} (CDP), incomparable to $(\eps,\delta)$-differential privacy but again having the same behavior under composition.  \Cdp{} is closer to the ``every~$\delta$'' property in the statement of Theorem~\ref{thm: advanced composition}:  An algorithm offers $(\mu,\tau)$-concentrated differential privacy if the privacy loss random variable has mean $\mu$ and if, after subtracting off this mean, the resulting (centered) random variable, $\xi$, is subgaussian with standard $\tau$.\footnote{A random variable $X$ is {\em subgaussian with standard $\sub$} for a constant $\sub > 0$ if $\forall \lambda \in \Rt: \ex[e^{\lambda \cdot X}] \leq e^{\frac{\lambda^2 \cdot \sub^2}{2}}$.  See Section~\ref{sec:defs:subgaussians}.}  In consequence (see, \eg , \cite{BuldyginK00} Lemma~1.3), $\forall x > 0$:
\begin{eqnarray*}
\Pr[\xi \ge x]\le \exp\left( - \frac{x^2}{2\tau^2} \right) ~~\mbox{\rm and}~~
\Pr[\xi \le -x]\le \exp\left( - \frac{x^2}{2\tau^2} \right)
\end{eqnarray*}
Thus, \cdp{} ensures that the expected privacy loss is $\mu$ and the probability of loss exceeding its mean by $x=t\tau$ is bounded by $e^{-t^2/2}$, echoing the form of the guarantee offered by the Advanced Composition Theorem (Theorem~\ref{thm: advanced composition}).

Consider the case in which $\tau = \eps$.  On the one hand, $(\mu,\eps)$-\cdp{} is clearly weaker than $(\eps,\delta)$-differential privacy, because even if the expected loss $\mu$ is very small, the probability of privacy loss exceeding $\eps$ in the former can be constant in the former but is only $\delta$, which is tiny, in the latter.  On the other hand, in $(\eps,\delta)$-differential privacy there is no bound on the expected privacy loss, since with probability $\delta$ all bets are off and the loss can be infinite.

\Cdp{} enjoys two advantages over $(\eps,\delta)$-differential privacy.
\begin{itemize}
\item
{\bf Improved Accuracy.} \Cdp{} is tailored to the (realistic!) case of large numbers of computations.  Traditionally, to ensure small cumulative loss with high probability, the permitted loss for each individual query is set very low, say $\eps' = \eps\sqrt{(\log 1/\delta)/k}$, even though a privacy loss of, say, $\eps/10$ or even $\eps$ itself may not be of great concern for any single query.   \gnote{should discuss why not have bigger $\delta$ in individual analyses (or larger $\eps$).} This is precisely the flexibility we give in \cdp : much less concern about single-query loss, but high probability bounds for cumulative loss.  The composition of $k$ $(\mu,\tau)$-\cdp{} mechanisms is $(k\mu, \sqrt{k}\tau)$-\cdp{} (Theorem~\ref{thm: composition of CDP}).  Setting $\tau = \eps$ we get an expected privacy loss of $k\mu$ and, for all $t$ simultaneously, the probability of privacy loss exceeding its epectation by $t\sqrt{k}\eps$ falls exponentially in $t^2/2$, just as we obtained in the composition for the other variants of differential privacy in Theorem~\ref{thm: advanced composition} above.  However, we get better accuracy.  For example, to handle $n$ counting queries using the Gaussian mechanism, we can add independent random noise drawn from ${\cal N}(0,n/\eps^2)$ to each query, achieving $(\eps(e^\eps - 1)/2,\eps)$-\cdp{} (Theorem~\ref{thm:Gauss-CDP})\footnote{To achieve $(\eps,\delta)$-differential privacy one adds noise drawn from ${\cal N}(0,2 \ln(1/\delta))$, increasing the typical distortion by a factor of $\sqrt{\ln(1/\delta)}$.}.  When $\eps = \Theta(1)$ the noise is scaled to $O(\sqrt{n})$; the Fundamental Law says noise~$o(\sqrt{n})$ is disastrous.

\item
{\bf Group Privacy}  {\em Group privacy} bounds the privacy loss even for pairs of databases that differ in the data of a small group of individuals; for example, in a health survey one may wish not only to know that one's own health information remains private but also that the information of one's family as a whole is strongly protected.  Any $(\eps,0)$-differentially private algorithm automatically ensures $(s\eps,0)$-differential privacy for all groups of size~$s$~\cite{DworkMNS06}, with the {\em expected} privacy loss growing by a factor of about $s^2$.  The situation for $(\eps,\delta)$-differential privacy is not quite so elegant: the literature shows $(s\eps,se^{s-1}\delta)$-differential privacy for groups of size~$s$, a troubling exponential increase in the failure probability (the $se^{s-1}\delta$ term).  The situation for \cdp{} is much better: for all known natural mechanisms with \cdp{} we get tight bounds. For (hypothetical) arbitrary algorithms offering subgaussian privacy loss, Theorem~\ref{thm:group-CDP} shows bounds that are asymptotically nearly-tight (tight up to low-order terms). We suspect that the tight bounds should hold for arbitrary mechanisms, and it would be interesting to close the gap.
\begin{enumerate}
\label{item:good}
\item Under certain conditions, satisfied by all pure differential privacy mechanisms\footnote{Every $(\eps,0)$-differentially private mechanism yields $(\eps(e^\eps-1)/2,\eps)$-\cdp{} (Theorem~\ref{thm:DP-to-CDP}).},
and the addition of Gaussian noise,
any $(\mu,\tau)$-\cdp{} mechanism satisfies $(s^2 \cdot \mu, s \cdot \tau)$-\cdp{} for groups of size~$s$, which is optimal.
\item
\label{item:arbitrary}
Every $(\frac{\tau^2}{2},\tau)$-\cdp{} mechanism satisfies $(s^2 \cdot \frac{\tau^2}{2} \cdot (1+o(1)), s \cdot \tau \cdot (1+o(1))$-\cdp{} for groups of size~$s$. The bound holds so long as $s \cdot \tau$ is small enough (roughly, $(1/\tau)$ should remain quasi-linear in $s$). See Theorem~\ref{thm:group-CDP} for the formal statement. We also assume here that $\mu \leq \frac{\tau^2}{2}$.\footnote{Up to low-order terms, this relationship holds for all mechanisms we know. For other mechanisms, it is possible to derive a less elegant general bound, or $\tau$ can be ``artificially inflated'' until this condition holds.}
\end{enumerate}
\end{itemize}

\begin{remark}
\label{rem: commute}
Consider any mechanism built by composing a collection of ``good'' mechanisms, each satisfying the conditions in Item~\ref{item:good} above.  To prove that the composed mechanism has the good group privacy bounds we can first apply the group privacy bounds for the underlying ``good'' mechanisms, and then apply composition.  It is interesting that these two operations commute.
\cnote{Check the math!!!}
\cnote{We should try to find a useful primitive satisfying \cdp{} but not satisfying our nice condition.}
\end{remark}

\noindent
{\bf Tight Bounds on Expected Loss.} As noted above, we improve by a factor of two the known upper bound on expected privacy loss of any $(\eps,0)$-differentially private mechanism, closing a gap open since~2010~\cite{DworkRV10}.  This immediately translates to an improvement by a factor of $\sqrt{2}$ on the utility/privacy tradeoff in {\em any} application of the Advanced Composition Theorem.  The new bound, which is tight, is obtained by first proving the result for special pairs of probablity distributions that we call {\em \MD} (privacy loss for any outcome is in $\{-\eps,0,\eps\}$), and then showing a reduction, in which an arbitrary pair of distributions with max divergence bounded by $\eps$ -- such as the distributions on outcomes of a differentially private mechanism when run on databases $\dba$ and $\dbb$ differing in a single row -- can be ``replaced'' by a \MD{} pair with no change in max divergence and no decrease in KL-divergence. \gnote{verify tightness, should be true for randomized response}

\begin{remark}
If all $(\eps,0)$-differentially private algorithms enjoy \cdp , as well as the Gaussian mechanism, which $(\eps,\delta)$-differentially private algorithms are ruled out?  All $(\eps,\delta)$-differentially private algorithms in which there is some probability $\delta' \le \delta$ of infinite privacy loss.  This includes many (but not all!) algorithms in the ``Propose-Test-Release'' framework~\cite{DworkL09}, in which a differentially private test is first performed to check that the dataset satisfies some ``safety'' conditions and, if so, an operation is carried out that only ensures privacy if the conditions are met.  There could be a small probability of failure in the first step, meaning that the test reports that the safety conditions are met, but in fact they are not, in which case privacy could be compromised in the second step.
\cnote{Add good example of infinite privacy loss in second step} \gnote{consensus?} \gnote{perhaps add that we can also allow an error parameter in CDP to allow these types of analyses}
\end{remark}


\subsection{Recent Developments}
\paragraph{Tightness.}
Optimality in privacy loss is a suprisingly subtle and difficult question under composition.  Results of Kairouz, Oh and Viswanath~\cite{KairouzOV15} obtain tight bounds under composition of arbitrary $(\epsilon,\delta)$-mechanisms, when all mechanisms share the same values of $\epsilon$ and $\delta$.  That is, they find the optimal $\epsilon',\delta'$ such that the composition of $k$ mechanisms, each of which is $(\epsilon_0,\delta_0)$-differentially private, is $(\epsilon',\delta')$-differentially private.  The nonhomogeneous case, in which the $i$th mechanism is $(\epsilon_i,\delta_i)$-differentially private, has been analyzed by Murtagh and Vadhan~\cite{MurtaghV16}, where it is shown that determining the bounds of the optimal composition is hard for $\# {\cal P}$.  Both these papers use an analysis similar to that found in our Lemma~\ref{lemma:KL-tight}, which we obtained prior to the publication of those works.
Optimal bounds on the composition of arbitrary mechanisms are very interesting, but we are also interested in bounds on the specific mechanisms that we have in hand and wish to use and analyze.  To this end, we obtain a complete characterization of the privacy loss of the Gaussian mechanism.  We show that the privacy loss of the composition of multiple application of the Gaussian mechanism, possibly with different individual parameters, is itself a Gaussian random variable, and we give exact bounds for its mean and variance. This characterization is not possible using the framework of $(\epsilon,\delta)$-differential privacy, the previous prevailing view of the Gaussian mechanism.

\paragraph{Subsequent Work.}
Motivated by our work, Bun and Steinke~\cite{BunS15} suggest a relaxation of concentrated differential privacy.  Instead of framing the privacy loss as a subgaussian random variable as we do here, they instead frame the question in terms of Renyi entropy, obtaining a relaxation of concentrated differential privacy that also supports a similar composition theorem. Their notion also provides privacy guarantees for groups. The bounds we get using concentrated differential privacy (Theorem \ref{thm:group-CDP}) are tighter; we do not know whether this is inherent in the definitions.

\section{Preliminaries}
\label{sec:defs}

\paragraph{Divergence.} We will need several different notions of divergence between distributions. We will also introduce a new notion, {\em subgaussian divergence} in Section \ref{sec:CDP}.

\begin{definition}[KL-Divergence]
\label{def:KL-divergence}

The KL-Divergence, or Relative entropy, between two random variables $Y$ and $Z$ is defined as:
$$\DKL(Y||Z) = \ex_{y \sim Y} \left[ \ln \frac{\Pr[Y=y]}{\Pr[Z=y]} \right],$$
where if the support of $Y$ is not equal to the support of $Z$, then $\DKL(Y||Z)$ is not defined.

\end{definition}

\begin{definition}[Max Divergence]
\label{def:Max-divergence}

The Max Divergence between two random variables $Y$ and $Z$ is defined to be:
$$\Dmax(Y||Z) = \max_{S \subseteq \Supp(Y)} \left[ \ln \frac{\Pr[Y \in S]}{\Pr[Z \in S]} \right],$$
where if the support of $Y$ is not equal to the support of $Z$, then $\Dmax(Y||Z)$ is not defined.

The $\delta$-approximate divergence between $Y$ and $Z$ is defined to be:
$$\Dapprox(Y||Z) = \max_{S \subseteq \Supp(Y): \Pr[Y \in S] \geq \delta} \left[ \ln \frac{\Pr[Y \in S]}{\Pr[Z \in S]} \right],$$
where if $\Pr[Y \in \Supp(Y) \setminus \Supp(Z)] > \delta$, then $\Dapprox(Y||Z)$ is not defined.
\end{definition}

\subsection{Differential Privacy}
\label{sec:defs:DP}

For a given database $\DB$, a
(randomized) non-interactive database access mechanism $\M$ computes
an output $\M(x)$ that can later be used to reconstruct information
about $\DB$. We will be concerned with mechanisms $\M$ that are {\em
private} according to various privacy notions described below.

We think of a database $\dba$ as a multiset of {\em rows}, each from
a data universe $U$. Intuitively, each row contains the data of a
single individual. We will often view a database of size $n$ as a
tuple $\dba\in U^n$ for some $n\in\N$ (the number of individuals
whose data is in the database). We treat $n$ as public information
throughout.

We say databases $\dba, \dbb$ are {\em adjacent} if they differ only
in one row, meaning that we can obtain one from the other by
deleting one row and adding another. I.e. databases are adjacent if
they are of the same size and their edit distance is 1.
To handle worst case pairs of databases, our probabilities
will be over the random choices made by the privacy mechanism.

\begin{definition} [$(\eps,0)$-Differential Privacy ($(\eps,0)$-DP)~\cite{DworkMNS06}]
\label{def: DP}

A randomized algorithm $\M$ is {\em $\eps$-differentially private}
if for all pairs of adjacent databases $\dba,\dbb$, and for all sets
$S \subseteq \Range(\M(\dba)) \cup \Range(\M(\dbb))$
$$\Pr[\M(\dba) \in S] \leq e^{\eps} \cdot Pr[\M(\dbb) \in S],$$
where the probabilities are over algorithm $\M$'s coins. Or alternatively: $$\Dmax(\M(\dba)||\M(\dbb)),\Dmax(\M(\dbb)||\M(\dba)) \leq \eps$$
\end{definition}

\begin{definition}[$(\eps, \delta)$-Differential Privacy ($(\eps,\delta)$-DP) \cite{DworkKMMN06}]
\label{def:delta-DP}

A randomized algorithm $\M$ gives
$(\eps,\delta)$-{\em differential privacy} if for all pairs of
adjacent databases $\dba$ and $\dbb$ and all $S \subseteq
\Range(\M)$
$$ \Pr[\M(\dba) \in S]  \le   e^\eps \cdot \Pr[\M(\dbb) \in S]
+ \delta, $$ where the probabilities are over the coin flips of the
algorithm $\M$. Or alternatively: $$\Dapprox(\M(\dba)||\M(\dbb)), \Dapprox(\M(\dbb)||\M(\dba)) \leq \eps$$
\end{definition}

\paragraph{Privacy Loss as a Random Variable.} 
Consider running an algorithm $\M$ on a pair of databases $x,y$. For an outcome $o$, the {\em privacy loss} on $o$ is the log-ratio of its probability when $\M$ is run on each database:
$$\Loss_{(\M(\dba)||\M(\dbb))}^{(o)} = \ln \frac{\Pr[\M(\dba)=o]}{\Pr[\M(\dbb)=o]} .$$

\Cdp{} delves more deeply into the privacy loss random variable: this real-valued random variable measures the privacy loss ensuing when algorithm $\M$ is run on $\dba$ (as opposed to $\dbb$). It is sampled by taking $y \sim \M(\dba)$ and outputting $\Loss_{(\M(\dba)||\M(\dbb))}^{(o)}$. This random variable can take positive or negative values. For $(\eps,0)$-differentially private algorithms, its magnitude is always bounded by $\eps$. For $(\eps,\delta)$-differentially private algorithms, with all but $\delta$ probability the magnitude is bounded by $\eps$.

\subsection{Subgassian Random Variables}
\label{sec:defs:subgaussians}

Subgaussian random variables were introduced by Kahane \cite{Kahane60}. A subgaussian random variable is one for which there is a positive real number $\sub > 0$ s.t. the moment generating function is always smaller than the moment generating function of a Gaussian with standard deviation $\sub$ and expectation~0.
In this section we briefly review the definition and basic lemmata from the literature.

\begin{definition}[Subgaussian Random Variable \cite{Kahane60}]
\label{def:subgaussian}

A random variable $X$ is {\em $\sub$-subgaussian} for a constant $\sub > 0$ if:
\begin{eqnarray*}
\forall \lambda \in \Rt: \ex[e^{\lambda \cdot X}] \leq e^{\frac{\lambda^2 \cdot \sub^2}{2}}
\end{eqnarray*}
We say that $X$ is subgaussian if there exists $\sub \geq 0$ s.t. $X$ is $\sub$-subgaussian.
For a subgaussian random variable $X$, the {\em subgaussian standard} of $X$ is:
\begin{eqnarray*}
\tau(X) = \mathit{inf} \{\sub \geq 0: \mbox{$X$ is $\sub$-subgaussian} \}
\end{eqnarray*}
\end{definition}
\paragraph{Remarks.}
An immediate consequence of Definition \ref{def:subgaussian} is that an $\sub$-subgaussian random variable has expectation 0, and variance bounded by $\sub^2$ (see Fact \ref{fact:subgauss-variance}). Note also that the gaussian distribution with expectation 0 and standard deviation $\sigma$ is $\sigma$-subgaussian. There are also known bounds on the higher moments of subgaussian random variables (see Fact \ref{fact:subgauss-moments}). See \cite{BuldyginK00} and \cite{Rivasplata12} for further discussion.

\begin{lemma}[Subgaussian Concentration]
\label{lemma:subgaussian-concentration}

If $X$ is $\sub$-subgaussian for $\sub > 0$, then:
\begin{eqnarray*}
\Pr[X \geq t \cdot \sub] \leq  e^{-t^2/2}, \hspace{1cm} \Pr[X \leq -t \cdot \sub]  \leq  e^{-t^2/2}
\end{eqnarray*}
\end{lemma}

\begin{proof}
For every $\lambda > 0$ and $t > 0$:
\begin{eqnarray*}
\Pr[X \geq t \cdot \sub] & = & \Pr[e^{\lambda \cdot X} \geq e^{\lambda \cdot t \cdot \sub}] \\
& \leq & e^{-\lambda \cdot t \cdot \sub} \cdot \ex[e^{\lambda \cdot X}] \\
& \leq & e^{\frac{\lambda^2 \cdot \sub^2}{2} - \lambda \cdot t \cdot \sub}
\end{eqnarray*}
where the first equality is obtained by raking an exponential of all arguments, the second (in)equality is by Markov, and the this is by the properties of the subgaussian random variable $X$. The right-hand side is minimized when $\lambda = t/\sub$, and thus we get that:
$$\Pr[X \geq t \cdot \sub] \leq e^{-t^2/2}$$
The proof that $\Pr[X \leq -t \cdot \sub] \leq e^{-t^2/2}$ is similar.
\end{proof}

\begin{fact}[Subgaussian Variance] \label{fact:subgauss-variance}
The variance of any $\tau$-subgaussian random variable $Y$ is bounded by $\Var(Y) \leq \tau^2$.
\end{fact}

\begin{fact}[Subgaussian Moments] \label{fact:subgauss-moments}
For any $\tau$-subgaussian random variable $Y$, and integer $k$, the $k$-th moment is bounded by:
$$\Ex[Y^{k}] \leq \left( (\lceil k/2 \rceil !) \cdot 2^{\lceil k/2 \rceil + 1} \cdot \tau^{k} \right)$$
\end{fact}

\begin{lemma}[Sum of Subgaussians]
\label{lemma:subgassian-composition}

Let $X_1,\ldots,X_k$ be (jointly distributed) real-valued random variables such that for every $i \in [k]$, and for every $(x_1,\ldots,x_{i-1}) \in \Supp(X_1,\ldots,X_{k-1})$, it holds that the random variable $(X_i|X_1=x_1,\ldots,X_{i-1}=x_{i-1})$ is $\sub_i$-subgaussian. Then the random variable $\sum_{i \in [k]} X_i$ is $\sub$-subgaussian, where $\sub = \sqrt{\sum_{i \in [k]} \sub_i^2}$.
\end{lemma}

\begin{proof}

The proof is by induction over $k$. The base case $k=1$ is immediate. For $k > 1$, for any $\lambda \in \Rt$, we have:
\begin{eqnarray*}
\ex[e^{\lambda \cdot  \sum_{i \in [k]} X_i} ]
& = & \ex_{(x_1,\ldots,x_{k-1}) \sim (X_1,\ldots,X_{k-1})) } \left[ \ex_{X_k} [e^{\lambda \cdot  \sum_{i \in [k]} X_i} | X_1=x_1,\ldots,X_{k-1}=x_{k-1}] \right] \\
& = & \ex_{(x_1,\ldots,x_{k-1}) \sim (X_1,\ldots,X_{k-1})) } \left[ e^{\lambda \cdot \sum_{i \in [k-1]} x_i } \cdot \ex_{X_k} [e^{\lambda \cdot  X_k} | X_1=x_1,\ldots,X_{k-1}=x_{k-1}] \right] \\
& \leq & \ex_{(x_1,\ldots,x_{k-1}) \sim (X_1,\ldots,X_{k-1})) } \left[ e^{\lambda \cdot \sum_{i \in [k-1]} x_i } \cdot e^{\frac{\lambda^2 \cdot \sub_k^2}{2}} \right] \\
& = & e^{\frac{\lambda^2 \cdot \sub_k^2}{2}} \cdot \ex[e^{\lambda \cdot  \sum_{i \in [k-1]} X_i} ] \\
& \leq &  e^{\frac{\lambda^2 \cdot \sub_k^2}{2}} \cdot  e^{\frac{\lambda^2 \cdot \sum_{i \in [k-1]} \sub_i^2}{2}} \\
& = & e^{\frac{\lambda^2 \cdot \sum_{i \in [k]} \sub_i^2}{2}}
\end{eqnarray*}
where the last inequality is by the induction hypothesis.
\end{proof}

The following technical Lemma about the products of (jointly distributed) random variables, one of which is exponential in a subgaussian, will be used extensively in proving group privacy:

\begin{lemma}[Expected Product with Exponential in Subgaussian] \label{lemma:subgauss-product}
Let $X$ and $Y$ be jointly distributed random variables, where $Y$ is $\sub$-Subgaussian for $\tau \leq 1/3$. Then:
\begin{align}
\Ex[X \cdot e^Y] &\leq \Ex[X] + \sqrt{\Ex \left[ X^2 \right]} \cdot (\tau + 3\tau^2)  \nonumber \\
&\leq \Ex[X] + (\sqrt{\Var(X)} + \Ex[X]) \cdot (\tau + 3\tau^2)
\end{align}
\end{lemma}

\begin{proof}
Taking the Taylor expansion of $e^Y$ we have:
\begin{align}
\Ex[X \cdot e^Y] &= \Ex \left[X \cdot \left( 1 + Y + \sum_{k=2}^{\infty} \frac{Y^k}{k!} \right) \right] \nonumber \\
&= \Ex[X] + \Ex[X \cdot Y] + \Ex \left[X \cdot \left(\sum_{k=2}^{\infty} \frac{Y^k}{k!} \right) \right] \label{eq:subgauss-product-1}
\end{align}
By the Cauchy-Schwartz inequality:

\begin{align}
\Ex[X \cdot Y] &\leq \sqrt{\Ex[X^2]} \cdot \sqrt{\Ex[Y^2]} \nonumber \\
&\leq  \sqrt{\Ex[X^2]} \cdot \tau  \label{eq:subgauss-product-2}
\end{align}
where the last inequality uses the fact that for a $\sub$-subgaussian RV $Y$, $\Var(Y) \leq \tau^2$ (Fact \ref{fact:subgauss-variance}), and that $\sqrt{a + b} \leq \sqrt{a} + \sqrt{b}$ (for any $a,b \geq 0$). To bound the last summand in Inequality \ref{eq:subgauss-product-1}, we use linearity of the expectation and Cauchy-Schwartz:

\begin{align}
\Ex \left[X \cdot \left(\sum_{k=2}^{\infty} \frac{Y^k}{k!} \right) \right] &= \sum_{k=2}^{\infty} \left( \Ex \left[X \cdot \frac{Y^k}{k!} \right] \right) \nonumber \\
&\leq \sum_{k=2}^{\infty} \sqrt{\Ex \left[ X^2 \right]} \cdot \sqrt{\Ex \left[ \frac{Y^{2k}}{(k!)^2} \right]} \nonumber \\
&= \sqrt{\Ex \left[ X^2 \right]} \cdot \sum_{k=2}^{\infty} \sqrt{\Ex \left[ \frac{Y^{2k}}{(k!)^2} \right]} \nonumber
\end{align}
Using the fact that for any $\tau$-subgaussian distribution $Y$, the $2k$-th moment $\Ex[Y^{2k}]$ is bounded by $\left( (k!) \cdot 2^{k+1} \cdot \tau^{2k} \right)$  (see Fact \ref{fact:subgauss-moments}), we conclude from the above that for $\tau < 1$:
\begin{align}
\Ex \left[X \cdot \left(\sum_{k=2}^{\infty} \frac{Y^k}{k!} \right) \right] &\leq \sqrt{\Ex \left[ X^2 \right]} \cdot \sum_{k=2}^{\infty} \sqrt{\frac{(k!) \cdot 2^{k+1} \cdot \tau^{2k}}{(k!)^2}} \nonumber \\
&= \sqrt{\Ex \left[ X^2 \right]} \cdot \sum_{k=2}^{\infty} \sqrt{ \frac{(k!) \cdot 2^{k+1} \cdot \tau^{2k}}{(k!)^2} } \nonumber \\
&=  \sqrt{\Ex \left[ X^2 \right]} \cdot  \sum_{k=2}^{\infty} \sqrt{ \frac{2^{k+1} \cdot \tau^{2k}}{(k!)} } \nonumber \\
&\leq \sqrt{\Ex \left[ X^2 \right]} \cdot \sum_{k=2}^{\infty} \sqrt{4\tau^{2k}} \nonumber \\
&= \sqrt{\Ex \left[ X^2 \right]} \cdot 2\tau^2 \cdot \sum_{k=0}^{\infty} \tau^k \nonumber \\
&= \sqrt{\Ex \left[ X^2 \right]} \cdot \frac{2\tau^2}{1- \tau} \label{eq:subgauss-product-3}
\end{align}
Putting together Inequalities \eqref{eq:subgauss-product-1}, \eqref{eq:subgauss-product-2}, \eqref{eq:subgauss-product-3},  we conclude that for $\tau \leq 1/2$:
\begin{align}
\Ex[X \cdot e^Y] &\leq  \Ex[X] + \left( \sqrt{\Ex \left[ X^2 \right]} \cdot \tau \right) +\left( \sqrt{\Ex \left[ X^2 \right]} \cdot \frac{2\tau^2}{1- \tau} \right) \nonumber \\
&= \Ex[X] + \sqrt{\Ex \left[ X^2 \right]} \cdot \left( \tau + \frac{2\tau^2}{1-\tau} \right) \nonumber \\
&\leq  \Ex[X] + \sqrt{\Ex \left[ X^2 \right]} \cdot (\tau + 3\tau^2)  \nonumber \\
&\leq  \Ex[X] + (\sqrt{\Var(X)} + \Ex[X]) \cdot (\tau + 3\tau^2)  \nonumber
\end{align}
where the next-to-last inequality holds whenever $\tau \leq 1/3$, and the last inequality is because $\sqrt{\Ex[X^2]} \leq \sqrt{\Var(X)} + \Ex[X]$ (because $\sqrt{a+b} \leq \sqrt{a} + \sqrt{b}$ for any $a,b>0$). \end{proof}

\section{Concentrated Differential Privacy: Definition and Properties}
\label{sec:CDP}

\begin{definition}[Privacy Loss Random Variable $\Loss_{(Y||Z)}$]
\label{def:privacy-loss}

For two discrete random variables $Y$ and $Z$, the {\em privacy loss random variable} $\Loss_{(Y||Z)}$, whose range is $\Rt$, is distributed by drawing $y \sim Y$, and outputting $\ln ({\Pr[Y=y]}/{\Pr[Z=y]})$. In particular, the expectation of $\Loss_{Y||Z}$ is equal to $\DKL(Y||Z)$.
If the supports of $Y$ and $Z$ aren't equal, then the privacy loss random variable is not defined.
\end{definition}

We study the privacy loss random variable, focusing on the case where this random variable is tightly concentrated around its expectation. In particular, we will be interested in the case where the privacy loss (shifted by its expectation) is subgaussian.

\begin{definition}[Subgaussian Divergence and Indistinguishability]
\label{def:SubGaus-divergence}

For two random variables $Y$ and $Z$, we say that $\DsubG(Y||Z) \preceq (\mu,\sub)$ if and only if:

\begin{enumerate}

\item
$\ex[\Loss_{(Y||Z)}] \leq \mu$

\item
The (centered) distribution $(\Loss_{(Y||Z)} - \ex[\Loss_{(Y||Z)}])$ is defined and subgaussian, and its subgaussian parameter is at most $\sub$.

\end{enumerate}

If we have both $\DsubG(Y||Z) \preceq (\mu,\sub)$ and $\DsubG(Z||Y) \preceq (\mu,\sub)$, then we say that the pair of random variables $X$ and $Y$ are {\em $(\mu,\sub)$-subgaussian-indistinguishable}.
\end{definition}

\begin{definition}[$(\mu,\sub)$-Concentrated Differential Privacy ($(\mu,\sub)$-CDP)]
\label{def:CDP}

A randomized algorithm $\M$ is {\em $(\mu,\sub)$-concentrated differentially private}
if for all pairs of adjacent databases $\dba,\dbb$, we have $\DsubG(\M(\dba)||\M(\dbb)) \preceq (\mu,\sub)$.
\end{definition}

\gnote{Eventually: add also $(\mu,\sub,\delta)$-CDP}

\begin{corollary}[Concentrated Privacy Loss]
\label{cor:CDP-concentration}

For every $(\mu,\sub)$-CDP algorithm $\M$, for all pairs of adjacent databases $\dba,\dbb$, taking $Y$ to be the distribution of $\M(\dba)$, and $Z$ to be the distribution of $\M(\dbb)$:
$$\Pr[\Loss_{(Y||Z)}  \geq \mu + (t \cdot \sub)] \leq \exp(-\frac{t^2}{2})$$

\end{corollary}

\begin{proof} Follows from Definition \ref{def:CDP} and the concentration properties of subgaussian random variables (Lemma \ref{lemma:subgaussian-concentration}). \end{proof}

\gnote{wrote as one-sided guarantee (seems fine). Can we improve by $1/t$?}

\paragraph{Guassian Mechanism Revisited.} We revisit the Gaussian noise mechanism of \cite{DworkMNS06}, giving a tight characterization of the privacy loss random variable.

\begin{theorem}[Gaussian is CDP]
\label{thm:Gauss-CDP}
Let $f: \dba \rightarrow \Rt$ be any real-valued function with sensitivity $\Delta(f)$. Then the Gaussian mechanism with noise magnitude $\sigma$ is $(\sub^2/2,\sub)$-CDP, where $\sub=\Delta(f)/\sigma$.
\end{theorem}

Later, we will prove (Theorem \ref{thm:DP-to-CDP}) that {\em every} pure differentially private mechanism also enjoys concentrated differential privacy.  The Gaussian mechanism is different, as it only ensures $(\eps,\delta)$-differential privacy for $\delta > 0$.

\begin{proof}[of Theorem~\ref{thm:Gauss-CDP}]
Let $\M$ be the Gaussian mechanism with noise magnitude $\sigma$. Let $\DB,\DB'$ be adjacent databases, and suppose w.l.o.g that $f(\DB) = f(\DB')+\Delta f$. We examine the privacy loss random variable obtained by drawing a noise magnitude $x \sim {\cal N}(0, \sigma^2)$ and outputting:
\begin{eqnarray*}
\ln \frac{\Pr[\M(\DB) = (f(\DB) + x)]}{\Pr[\M(\DB')=(f(\DB) + x)]} & = & \ln \frac{e^{(-1/2\sigma^2) \cdot x^2} }{e^{(-1/2\sigma^2) \cdot (x + { \Delta f})^2}} \\
&=& \ln e^{(-1/2\sigma^2) \cdot \left[ x^2 - (x + \Delta f)^2 \right]}\\
&=& - \frac{1}{2\sigma^2} \cdot \left[ x^2 - (x^2 + 2x \cdot \Delta f + (\Delta f)^2) \right] \\
&=& - \frac{1}{2\sigma^2} \cdot \left[ -2x \cdot \Delta f - (\Delta f)^2 \right] \\
&=& \left( \frac{\Delta f}{\sigma} \cdot \frac{x}{\sigma} \right) + \frac{1}{2} \left( \frac{\Delta f}{\sigma} \right)^2
\end{eqnarray*}

Since $x \sim {\cal N}(0, \sigma^2)$, we conclude that the distribution of the privacy loss random variable $\Loss_{(U||V)}$ is also Gaussian, with expectation $\left( {\Delta f}/{\sigma} \right)^2 / 2$, and standard deviation ${\Delta f}/{\sigma}$. Taking $\sub = {\Delta f}/{\sigma}$, we get that $\DsubG(\M(\DB)||\M(\DB')) \preceq (\sub^2/2, \sub)$.
\end{proof}

As noted in the Introduction, it is a consequence of Theorem~\ref{thm:Gauss-CDP} that we can achieve $(\eps(e^\eps - 1)/2,\eps)$-\cdp{} by adding independent random noise drawn from ${\cal N}(0,n/\eps^2)$ to each query.  If we further relax to $(\eps,\sqrt{\eps})$-cdp, we can add even smaller noise, of magnitude $(\sqrt{1/\eps})$. This would make sense in settings where we expect further composition, and so we can focus on the expected privacy loss (bounding it by $\eps$) and allow more slackness in the standard deviation. For small $\eps$, this gives another order of magnitude improvement in the amount of distortion introduced to protect privacy.

Finally, we observe that the bounds for group \cdp{} of the Gaussian mechanism follow immediately from from Theorem \ref{thm:Gauss-CDP}, noting that for a group of size~$s$ the group sensitivity of a function~$f$ is at most $s \cdot \Delta f$.

\begin{corollary}[Group CDP for the Gaussian Mechanism]
The Gaussian mechanism with noise magnitude $\sigma$ satisfies $((s\Delta f/\sigma)^2/2,s\Delta f)$-\cdp .
\end{corollary}

\subsection{Composition}
Concentrated Differential Privacy composes ``as well as'' standard differential privacy. Indeed, a primary advantage of CDP is that it permits greater accuracy and smaller noise, with essentially no loss in privacy under composition. In this section we prove these composition properties. We follow the formalization in \cite{DworkRV10} \gnote{missing reference: DworkJ12} in modeling composition. Composition covers both repeated use of (various) CDP algorithms on the same database, which allows modular construction of CDP algorithms, and repeated use of (various) CDP algorithms on different databases that might contain information pertaining to the same individual. In both of these scenarios, the improved accuracy of CDP algorithms can provide greater utility for the same ``privacy budget''.

Composition of $k$ CDP mechanisms (over the same database, or different databases) is formalized by a sequence of pairs of random variables $(U,V)= ((U^{(1)},V^{(1)}),\ldots,(U^{(k)},V^{(k)}))$. The random variables are the outcomes of adversarially and adaptively chosen CDP mechanisms $\M_1,\ldots,\M_k$. In the $U$ sequence (reality), the random variable $U^{(i)}$ is sampled by running mechanism $\M_i$ on a database (of the adversary's choice) containing an individual's, say Bob's, data. In the $V$ sequence (alternative reality), the random variable $V^{(i)}$ is sampled by running mechanism $\M_i$ on the same database, but where Bob's data are replaced by (adversarially chosen) data belonging to a different individual, Alice. The requirement is that even for adaptively and adversarially chosen mechanisms and database-pairs, the outcome of $U$ (Bob-reality) and $V$ (Alice-reality) are ``very close'', and in particular the privacy loss $\Loss_{(U||V)}$ is subgaussian.

In more detail, and following~\cite{DworkRV10}, we define a game in which a dealer flips a fair coin to choose between symbols $U$ and $V$, and an adversary adaptively chooses a sequence of pairs of adjacent databases $(x_i^U,x_i^V)$ and a mechanism $\M_i$ enjoying $(\mu_i,\sub_i)$-CDP and that will operate on either the left element (if the dealer chose~$U$) or the right element (if the dealer chose~$V$) of the pair, and return the output, for $1 \le i \le k$.  The adversary's choices are completely adaptive and hence may depend not only on arbitrary external knowledge but also on what has been observed in steps $1, \dots, i-1$.  The goal of the adversary is to maximize privacy loss.  It is framed as a game because large privacy loss is associated with an increased ability to determine which of $(U,V)$ was selected by the dealer, and we imagine this to be the motivation of the adversary.


\begin{theorem}[Composition of CDP]
\label{thm: composition of CDP}
For every integer $k \in \Nt$, every $\mu_1, \dots, \mu_k, \sub_1, \dots, \sub_k \geq 0$, and $$(U,V) = ((U^{(1)},V^{(1)}),\ldots,(U^{(k)},V^{(k)}))$$ constructed as in the game described above, we have that $\DsubG(U||V) \preceq (\sum_{i=i}^k \mu_i, (\sum_{i=1}^k \sub_i^2)^{1/2})$.
\end{theorem}

\begin{proof}
Consider the random variables $U$ and $V$ defined above, and the privacy loss random variable $\Loss_{(U||V)}$. This random variable is obtained by picking $\by \sim U$ and outputting $\ln \frac{\Pr[U=\by]}{\Pr[V=\by]}$.

The mechanism and datasets chosen by the adversary at step~$i$ depend on the adversary's view at that time.  The adversary's view comprises its randomness and the outcomes it has observed thus far.  Letting $R_U$ and $R_V$ denote the randomness in the $U$-world and $V$-world respectively, we have, for any $\by = (y_1,\ldots,y_k) \in \Supp(U)$ and random string $r$ 
\begin{eqnarray*}
\ln \frac{\Pr[U=\by]}{\Pr[V=\by]} & = & \ln \left(
\frac{\Pr[R_U = r]}{\Pr[R_V = r]} \cdot
 \frac{\prod_{i \in [k]} \Pr[U^{(i)}=y_i | U^{(i-1)}=y_{i-1},\ldots, U^{(1)}=y_1]}{\prod_{i \in [k]} \Pr[V^{(i)}=y_i | V^{(i-1)}=y_{i-1},\ldots, V^{(1)}=y_1]} \right) \\
& = & \sum_{i \in [k]} \ln \frac{\Pr[U^{(i)}=y_i|U^{(i-1)}=y_{i-1},\ldots, U^{(1)}=y_1]}{\Pr[V^{(i)}=y_i|V^{(i-1)}=y_{i-1},\ldots, V^{(1)}=y_1]} \\
 & \triangleq &  \sum_{i \in [k]} c_i(r,y_1,\ldots,y_i) \, .
\end{eqnarray*}

Now for every prefix $(r,y_1,\ldots,y_{i-1})$ we condition on $R_U=r,U_1=y_1,\ldots,U_{i-1}=y_{i-1}$, and analyze the the random variable $c_i(R_U,U_1,\ldots,U_i)=c_i(r,y_1,\ldots,y_{i-1},U^{(i)}$. Once the prefix is fixed, the next pair of databases $x_i^U$ and $x_i^V$ and the mechanism $\M_i$ 
output by the advesary are also determined. Thus $U_i$ is distributed according to $\M_i(x_i^U)$ and for any value $y_i$, we have $$c_i (r,y_1,\ldots,y_{i-1},y_i)
= \ln\left(\frac{\Pr[\M_i(x_i^U)=y_i]}{\Pr[\M_i(x_i^V)=y_i]}\right)$$ which is simply the privacy loss when $\M_i(x_i^U)=y_i$.  By the CDP properties of $\M_i$, $\Loss_{(\M_i(x_i^U) || \M_i(x_i^V))}$ is $(\mu_i,\sub_i)$ subgaussian.

By the subgaussian properties of the random variables $C_i = c_i(r,U^{(1)},\ldots,U^{(i)})$, we have that $\Loss_{(U||V)} = \sum_{i \in [k]} C_i$, i.e. the privacy loss random variable equals the sum of the $C_i$ random variables.  By linearity of expectation, we conclude that:
\begin{equation*}
\ex[\Loss_{(U||V)}] = \ex[\sum_{i \in [k]} C_i] = \sum_{i \in [k]} \ex[C_i]  = \sum_{i in [k]} \mu_i \,
\end{equation*}
and by Lemma \ref{lemma:subgassian-composition}, we have that the random variable:
\begin{equation*}
(\Loss_{(U||V)} - \ex[\Loss_{(U||V)}]) = \sum_{i \in [k]} (C_i - \ex[C_i])
\end{equation*}
is $\left(\sum_{i \in [k]} \tau_i^2 \right)^{1/2}$-subgaussian.
\end{proof}


\subsection{Relationship to DP}
\label{sec: relationship to DP}

In this section, we explore the relationship between differential privacy and concentrated differential. We show that {\em any} differentially private algorithm is also concentrated differentially private. Our main contribution here is a refined upper bound on the {\em expected} privacy loss of differentially private algorithms: we show that if $\M$ is $\eps$-DP, then its expected privacy loss is only (roughly) $\eps^2/2$ (for small enough $\eps$). We also show that the privacy loss random variable for {\em any} $\eps$-DP algorithm is subgaussian, with parameter $\sub=O(\eps)$:

\begin{theorem}[DP $\Rightarrow$ CDP]
\label{thm:DP-to-CDP}
Let $\M$ be any $\eps$-DP algorithm. Then $\M$ is $(\eps \cdot (e^{\eps}-1)/2 , \eps)$-CDP.
\end{theorem}
\gnote{eventually, add theorem for $(\eps,\delta)$-DP}
\cnote{What will it say?}

\begin{proof} 
  Since $\M$ is $(\eps,0)$-differentially private, we know that the privacy loss random variable is always bounded in magnitude by $\eps$.  The random variable obtained by subtracting off the expected privacy loss, call it $\mu$, therefore has mean zero and lies in the interval $[-\eps-\mu,\eps-\mu]$.  It follows from Hoeffding's Lemma, stated next, that such a bounded, centered, random variable is $(\eps-\mu - (-\eps-\mu))/2 = \eps$-subgaussian.
\begin{lemma}[Hoeffding's Lemma]
\label{lem: Hoeffding}
Let $X$ be a zero-mean random variable such that $\Pr[X \in [a,b]] = 1$.
Then $\E[e^{\lambda X}] \le e^{(1/8)\lambda^2(b-a)^2}$.
\end{lemma}

The main challenge is therefore to bound the expectation, namely the quantity $\DKL(D||D')$, where $D$ is the distribution of $\M(\dba)$ and $D'$ is the distribution of $\M(\dbb)$, and $\dba,\dbb$ are adjacent databases. In \cite{DworkRV10} it was shown that:
\begin{lemma}[\cite{DworkRV10}]
\label{lemma:DRV}
For any two distributions $D$ and $D'$ such that $\Dmax(D||D'),\Dmax(D'||D) \leq \eps$,
$$\DKL(D||D') \leq \DKL(D||D') + \DKL(D'||D) \leq \eps \cdot (e^{\eps} - 1)$$
\end{lemma}

We improve this bound, obtaining the following refinement:
\begin{lemma}
\label{lemma:KL-tight}
For any two distributions $D$ and $D'$ such that $\Delta_{\infty}(D,D') = \eps$,
$$\DKL(D||D') \leq \eps \cdot (e^{\eps} - 1) / 2$$
\end{lemma}

The proof of Theorem \ref{thm:DP-to-CDP} follows from Lemma \ref{lemma:KL-tight}. To prove Lemma \ref{lemma:KL-tight}, we introduce the notion of \MD{} distributions:

\begin{definition}[\MDcaps {} Distributions]
\label{def:max-divergent}

Let $D$ and $D'$ be two distributions with support $\X$, such that $\Delta_{\infty}(D,D') \leq \eps$ for some $\eps > 0$. We say that $D$ and $D'$ are {\em \MD} if $\forall x \in \X, \ln \frac{D[x]}{D'[x]} \in \{-\eps,0,\eps\}$.
\end{definition}

We then use the following two lemmas about maximally divergent distributions (the proofs follow below) to prove Lemma \ref{lemma:KL-tight}:

\begin{lemma}
\label{lemma:max-divergent}

For any two distributions $D$ and $D'$, there exist {\em \MD} distributions $M$ and $M'$ such that $\Delta_{\infty}(M,M') = \Delta_{\infty}(D,D')$ and $\DKL(D,D') \leq \DKL(M,M')$. Note that the support of $D,D'$ may differ from the support of $M,M'$.
\end{lemma}

\begin{lemma}
\label{lemma:KL-diff}

For any \MD{} distributions $M$ and $M'$, as in Definition \ref{def:max-divergent}, $\DKL(M,M')=\DKL(M',M)$.
\end{lemma}

\begin{proof}[Proof of Lemma \ref{lemma:KL-tight}]
By Lemma \ref{lemma:max-divergent} there exist \MD{} distributions $M$ and $M'$ s.t. $\Delta_{\infty}(M,M') \leq \eps$ and $\DKL(D,D') \leq \DKL(M,M')$. By lemma \ref{lemma:DRV}, $\DKL(M||M') + \DKL(M'||M) \leq \eps \cdot (e^{\eps} - 1)$. By Lemma \ref{lemma:KL-diff}, $\DKL(M||M') = \DKL(M'||M)$, and so $\DKL(M,M') \leq \eps \cdot (e^{\eps} - 1)/2$. Putting these together:
$$\DKL(D,D') \leq \DKL(M,M') \leq \eps \cdot (e^{\eps} - 1)/2$$
\end{proof}

\begin{proof}[Proof of Lemma \ref{lemma:max-divergent}]

Let $\eps = \Delta_{\infty}(D,D')$. We construct $M$ and $M'$ iteratively from $D$ and $D'$ by enumerating over each $x \in \X$. For each such $x$, we add a new element $s_x$ to the support. The idea is that the mass of $x$ in $D$ and $D'$ will be ``split'' between $x$ and $s_x$ in $M$ and $M'$ (respectively). This split will ensure that the probabilities of $s_x$ in $M$ and in $M'$ are identical, and the probabilities of $x$ in $M$ and $M'$ are ``maximally divergent''. We will show that the ``contribution'' of $x$ and $s_x$ to the KL divergence from $M$ to $M'$ is at least as large as the contribution of $x$ to the KL divergence from $D$ to $D'$. The lemma then follows.

We proceed with the full specification of this ``split'' and then formalize the above intuition. For $x \in \X$, take $p_x = D'[x]$. Since $\Delta_{\infty}(D,D') = e^{\eps}$, there must exist $\alpha \in [-1,1]$ s.t. $D[x] = e^{\alpha \cdot \eps} \cdot p_x$. We introduce a new item $s_x$ into the support, and set the mass of $M$ and $M'$ on $x$ and $s_x$ as follows:
\begin{eqnarray*}
M'[x] & = & p_x \cdot \frac{e^{\alpha \cdot \eps}-1}{e^{\signa \cdot \eps} - 1} \\
M[x] & = & e^{\signa \cdot \eps} \cdot M[x] \\
    & = & e^{\signa \cdot \eps} \cdot p_x \cdot \frac{e^{\alpha \cdot \eps}-1}{e^{\signa \cdot \eps} - 1} \\
M'[s_x] & = & D'[x] - M'[x] \\
 & = & p_x \cdot (1 - \frac{e^{\alpha \cdot \eps}-1}{e^{\signa \cdot \eps} - 1}) \\
M[s_x] & = & D[x] - M[x] \\
 & = & p_x \cdot (e^{\alpha \cdot \eps} - e^{\signa \cdot \eps} \cdot \frac{e^{\alpha \cdot \eps}-1}{e^{\signa \cdot \eps} - 1}) \\
 & = & p_x \cdot ( e^{\alpha \cdot \eps} - (e^{\signa \cdot \eps} - 1 + 1) \cdot \frac{e^{\alpha \cdot \eps}-1}{e^{\signa \cdot \eps} - 1}) \\
 & = & p_x \cdot (e^{\alpha \cdot \eps} - (e^{\signa \cdot \eps} - 1) \cdot \frac{e^{\alpha \cdot \eps}-1}{e^{\signa \cdot \eps} - 1} - \cdot \frac{e^{\alpha \cdot \eps}-1}{e^{\signa \cdot \eps} - 1}) \\
 & = & p_x \cdot (e^{\alpha \cdot \eps} - (e^{\alpha \cdot \eps}-1) - \frac{e^{\alpha \cdot \eps}-1}{e^{\signa \cdot \eps} - 1}) \\
 & = & p_x \cdot (1 - \frac{e^{\alpha \cdot \eps}-1}{e^{\signa \cdot \eps} - 1}) \\
 & = & M'[s_x]
\end{eqnarray*}
Observe that all probabilities are at least zero and sum to 1 (for $M$ and $M'$ respectively). Moreover, $M$ and $M'$ are \MD , $\Delta_{\infty}(M,M')={\eps}=\Delta_{\infty}(D,D')$. Thus, the distributions $M$ and $M'$ satisfy the conditions of the lemma. Finally, we emphasize that by the above we have $\forall x \in X, M[s_x]=M'[s_x]$.

We now compare the KL divergence from $D$ to $D'$ with the divergence from $M$ to $M'$.
\begin{eqnarray*}
\DKL(M,M') - \DKL(D,D') & = &
\sum_{x \in \X} \big[ \big( M[x] \cdot \ln \frac{M[x]}{M'[x]} - D[x] \cdot \ln \frac{D[x]}{D'[x]} \big) + \big( M[s_x] \cdot \ln \frac{M[s_x]}{M'[s_x]} \big) \big] \\
& = & \sum_{x \in \X} \big( M[x] \cdot \ln \frac{M[x]}{M'[x]} - D[x] \cdot \ln \frac{D[x]}{D'[x]} \big)
\end{eqnarray*}

Proposition \ref{prop:KL-grows} below, shows that for every $x \in \X$, the summand of $x$ in the above sum is non-negative. The lemma follows. \end{proof}

\begin{proposition}
\label{prop:KL-grows}
For every $x \in \X$:
$$\big( M[x] \cdot \ln \frac{M[x]}{M'[x]} - D[x] \cdot \ln \frac{D[x]}{D'[x]} \big) \geq 0$$
\end{proposition}

\begin{proof}
We use the following equality, and the complete the proof using a case analysis, depending on whether $\alpha \geq 0$ or $\alpha < 0$.
\begin{eqnarray*}
M[x] \cdot \ln \frac{M[x]}{M'[x]} - D[x] \cdot \ln \frac{D[x]}{D'[x]} & = &  (p_x \cdot e^{\signa \cdot \eps} \cdot \frac{e^{\alpha \cdot \eps}-1}{e^{\signa \cdot \eps} - 1} \cdot \signa \cdot \eps) - (p_x \cdot e^{\alpha \cdot \eps} \cdot (\alpha \cdot \eps)) \\
& = & p_x \cdot \eps \cdot \big(\signa \cdot  e^{\signa \cdot \eps} \cdot \frac{e^{\alpha \cdot \eps}-1}{e^{\signa \cdot \eps} - 1} - e^{\alpha \cdot \eps} \cdot \alpha \big) \\
& = & \frac{p_x \cdot \eps}{e^{\signa \cdot \eps}-1} \cdot \big((\signa \cdot  e^{\signa \cdot \eps} \cdot (e^{\alpha \cdot \eps} - 1) - ((e^{\signa \cdot \eps} - 1) \cdot e^{\alpha \cdot \eps} \cdot \alpha) \big) \\
& = & \frac{p_x \cdot \eps \cdot e^{\signa \cdot  \eps}}{e^{\signa \cdot  \eps}-1} \cdot \big( \signa \cdot  (e^{\alpha \cdot \eps} - 1)  -  ( e^{\alpha \cdot \eps} \cdot \alpha)  + (e^{-\signa \cdot \eps} \cdot e^{\alpha \cdot \eps} \cdot \alpha ) \big) \\
& = & \frac{p_x \cdot \eps \cdot e^{\signa \cdot  \eps} }{e^{\signa \cdot  \eps}-1} \cdot \big( e^{\alpha \cdot \eps} \cdot (\signa + \alpha \cdot (-1 + e^{-\signa \cdot  \eps}) ) - \signa \big)
\end{eqnarray*}
We proceed with a case analysis:

\paragraph{Case I: $\alpha \geq 0$.} Here $\sign(\alpha)=1$. We use the following inequality:
\begin{claim}
For $\eps \geq 0, \alpha \in [0,1]$:
\begin{eqnarray*}
\alpha \cdot (-1 + e^{-\eps}) \geq e^{-\alpha \cdot \eps} - 1
\end{eqnarray*}
\end{claim}

\begin{proof}
Observe that, for any fixed $\alpha \in [0,1]$, we have equality when $\eps=0$ (both the left-hand and right-hand sides of the above inequality equal 0). Taking derivatives (by $\eps$) for both sides, on the left-hand side the derivative is $-\alpha \cdot e^{-\eps}$, whereas on the right-hand side it is $-\alpha \cdot e^{-\alpha \cdot \eps}$. We conclude that $\forall \eps > 0,\alpha \in [0,1]$, the derivative left-hand side is at least the derivative on the right-hand side (because in this range $-\alpha \cdot e^{-\eps} \geq -\alpha \cdot e^{-\alpha \cdot \eps}$).\end{proof}

Now, by the above we have:
\begin{eqnarray*}
M[x] \cdot \ln \frac{M[x]}{M'[x]} - D[x] \cdot \ln \frac{D[x]}{D'[x]} & = & \frac{p_x \cdot \eps \cdot e^{\eps} }{e^{\eps}-1} \cdot \big( e^{\alpha \cdot \eps} \cdot (1 + \alpha \cdot (-1 + e^{-\eps}) ) - 1 \big) \\
& \geq & \frac{p_x \cdot \eps \cdot e^{\eps} }{e^{\eps}-1} \cdot \big( e^{\alpha \cdot \eps} \cdot (1 + (e^{-\alpha \cdot \eps} - 1)) - 1 \big) \\
& = & 0
\end{eqnarray*}

\paragraph{Case II: $\alpha < 0$.} Here $\sign(\alpha)=-1$. We use the following inequality:

\begin{claim}
\label{claim:negative-case}
For $\eps \geq 0, \alpha \in [-1,0)$:
\begin{eqnarray*}
\alpha \cdot (-1 + e^{\eps}) & \leq & -e^{-\alpha \cdot \eps} + 1
\end{eqnarray*}
\end{claim}

\begin{proof}
Observe that, for any fixed $\alpha \in [-1,0)$, we have equality when $\eps=0$ (both the left-hand and right-hand sides of the above inequality equal 0). Taking derivatives (by $\eps$) for both sides, on the left-hand side the derivative is $\alpha \cdot e^{\eps}$, whereas on the right-hand side it is $\alpha \cdot e^{-\alpha \cdot \eps}$. We conclude that $\forall \eps \geq 0,\alpha \in [-1,0)$, the derivative on the left-hand side is always at most the derivative on the right-hand side (because in this range $\alpha \cdot e^{\eps} \leq \alpha \cdot e^{-\alpha \cdot \eps}$).\end{proof}

Recall from above that:
\begin{eqnarray*}
M[x] \cdot \ln \frac{M[x]}{M'[x]} - D[x] \cdot \ln \frac{D[x]}{D'[x]} & = &
\frac{p_x \cdot \eps \cdot e^{-\eps} }{e^{-\eps}-1} \cdot \big( e^{\alpha \cdot \eps} \cdot (-1 + \alpha \cdot (-1 + e^{\eps}) ) +1 \big)
\end{eqnarray*}
The right-hand side of this equation is a product of two terms. The first is:
\begin{eqnarray*}
\frac{p_x \cdot \eps \cdot e^{-\eps} }{e^{-\eps}-1} \leq 0
\end{eqnarray*}
(because the numerator is positive, and the denominator is negative).

For the second term in the product, we use Claim \ref{claim:negative-case}, and get:
\begin{eqnarray*}
( e^{\alpha \cdot \eps} \cdot (-1 + \alpha \cdot (-1 + e^{\eps}) ) +1) & \leq & ( e^{\alpha \cdot \eps} \cdot (-1 + -e^{-\alpha \cdot \eps} + 1 ) +1) \\
& = & 0
\end{eqnarray*}
We conclude that in this case ($\alpha < 0$), the difference in $x$'s contribution to the KL divergences equals the product of two non-positive terms, and so it must be non-negative.
\end{proof}

\begin{proof}[Proof of Lemma \ref{lemma:KL-diff}]

Let $\eps = \Delta_{\infty}(M,M')$. Since $M$ and $M'$ are \MD , w.l.o.g we can study their KL divergence for the special case where the support is over 3 items: $\{x,y,s\}$, and where
$$M[x]=p, M[y]=q, M[s]=r$$
$$M'[x]=p'=p \cdot e^{-\eps}, M'[y] =q'=q \cdot e^{\eps}, D'[s]=r'=r$$

We analyze the expected privacy losses, or KL divergences, from $M$ to $M'$ and from $M'$ to $M$. First, observe that $1 = (p + q + r) = (p' + q' + r')$. We conclude that:
\begin{eqnarray*}
p+q = p' +q' & \Rightarrow & p+q = p \cdot e^{-\eps} + q \cdot e^{\eps} \\
& \Rightarrow & p \cdot (1 - e^{-\eps}) = q \cdot (e^{\eps} - 1) \\
& \Rightarrow & p= q \cdot \frac{e^{\eps} - 1}{1 - e^{-\eps}}
\end{eqnarray*}
Examining the KL divergences, we have:
\begin{eqnarray*}
\DKL(M||M') & = & p_x \cdot \ln(p/p') + q \cdot \ln (q/q') + s \cdot \ln (s/s') \\
& = & p_x \cdot \eps - q \cdot \eps \\
& = & \eps \cdot (p_x- q)\\
\DKL(M'||M) & = & p' \cdot \ln(p'/p) + q' \cdot \ln (q'/q) + s' \cdot \ln (s'/s) \\
& = & - p_x \cdot e^{-\eps} \cdot \eps + q \cdot e^{\eps} \cdot \eps \\
& = & \eps \cdot (q \cdot e^{\eps} - p_x \cdot e^{-\eps})
\end{eqnarray*}

We bound the difference as follows:
\begin{eqnarray*}
\DKL(M||M') - \DKL(M'||M) & = & \eps \cdot (p- q) - \eps \cdot (q \cdot e^{\eps} - p \cdot e^{-\eps}) \\
& = & \eps \cdot \big( p \cdot (e^{-\eps} + 1) - q \cdot (e^{\eps} + 1) \big) \\
& = & \eps \cdot q \cdot \big( \frac{e^{\eps} - 1}{1 - e^{-\eps}} \cdot (e^{-\eps} + 1) -  (e^{\eps} + 1) \big) \\
& = & \eps \cdot q \cdot \big( \frac{e^{\eps} - 1}{1 - e^{-\eps}} \cdot (e^{-\eps} - 1 + 2) -  (e^{\eps} + 1) \big) \\
& = & \eps \cdot q \cdot \big( \frac{(e^{\eps} - 1) \cdot (e^{-\eps} - 1)}{1 - e^{-\eps}} + 2 \cdot \frac{e^{\eps} - 1}{1 - e^{-\eps}} -  (e^{\eps} + 1) \big) \\
& = & \eps \cdot q \cdot \big( -(e^{\eps} - 1) + 2 \cdot \frac{e^{\eps} - 1}{1 - e^{-\eps}} -  (e^{\eps} + 1)) \\
& = & 2 \cdot \eps \cdot q \cdot \big( \frac{e^{\eps} - 1}{1 - e^{-\eps}} - e^{\eps} \big) \\
& = & 2 \cdot \eps \cdot q \cdot \big( \frac{(e^{\eps} - 1) - e^{\eps} \cdot (1 - e^{-\eps})}{1 - e^{-\eps}} \big) \\
& = & 2 \cdot \eps \cdot q \cdot \big( \frac{(e^{\eps} - 1) - e^{\eps} + 1}{1 - e^{-\eps}} \big)\\
& = & 0
\end{eqnarray*}
\end{proof}
\end{proof}

\section{Group Privacy}
\label{sec:group-CDP}

We show that arbitrary mechanisms that guarantee \cdp{} also provide \cdp{} for groups. This is stated in Theorem \ref{thm:group-CDP} below. The bounds are asymptotically nearly-tight, up to low-order terms. It would be interesting to tighten these bounds to match the tight group privacy guarantees of (all) known \cdp{} mechanisms, such as the Gaussian Mechanism or pure-$\eps$ Differentially private mechanisms. See the discussion in the introduction.

\begin{theorem}[Group CDP] \label{thm:group-CDP}
Let $\M$ be a {\em $(\mu,\sub)$-concentrated differentially private} mechanism. Let $\dba,\dbb$  be a pair of databases that differ on exactly $s$ rows. Suppose that $(\sub \cdot s \cdot \log^3 s)$ is bounded from above by a sufficiently small constant and $\mu \leq \sub^2/2$. Then:
\begin{align}
\DsubG(\M(\dba)||\M(\dbb)) & \preceq \left( \frac{(s \cdot \tau)^2}{2} + \tilde{O}((s \cdot \tau)^{2.5}) , (s \cdot \tau) + \tilde{O}((s \cdot \tau)^{1.5}) \right) \nonumber
\end{align}
(Note that since $(s \cdot \tau) < 1$, this implies that the privacy loss random variable has expectation roughly $\frac{(s \cdot \tau)^2}{2}$, and the subgaussian standard is roughly $(s \cdot \tau)$, all up to the low-order terms).
\end{theorem}

\begin{proof}

We assume for convenience that $s$ is a power of 2. The proof will be by induction over the value of $\log s$. For $s=1$ the claim follows immediately. For the induction step, suppose that the claim is true for databases differing on $2^{m}$ rows. We will show that it is true for $\dba,\dbb$ that differ on $2^{m+1}$ rows.  We take $\mu_m$ and $\tau_m$ to be bounds on the expectation and the standard of the (centered) privacy loss distribution for databases that differ on at most $2^m$ rows, so $\mu_0 = \mu$ and $\tau_0 = \tau$. We maintain the invariant that $\mu_m \leq \tau_m^2/2$ (for the base case $m=0$ this holds by the lemma conditions).

For the induction step, let $\dbc$ be an ``midpoint'' database that differs from both $\dba$ and $\dbb$ on exactly $2^{m}$ rows. Define the mechanism's output distributions on these databases by $D = \M(\dba), D'=\M(\dbc), D''=\M(\dbb)$. By the induction hypothesis, we conclude that:
\begin{align*}
\DsubG(D||D'),\DsubG(D'||D) & \preceq (\mu_{m},\tau_{m}) \\
\DsubG(D''||D'),\DsubG(D'||D'')  & \preceq (\mu_{m},\tau_{m})
\end{align*}
We use Lemmas \ref{lemma:group-expectation} and \ref{lemma:group-standard}, stated and proved in Sections \ref{subsec:group-expectation} and \ref{subsec:group-standard} below, to bound $\mu_{m+1}$ and $\tau_{m+1}$.

\paragraph{Bounding $\tau_{m+1}$.} By Lemma \ref{lemma:group-standard}, we have that
\begin{align}
\tau_{m+1} &\leq 2 \tau_{m} + 34\tau_m^{1.5}  \label{eq:group-recursive}
\end{align}
We prove that this recursive relation implies that:
\begin{align}
\tau_{m+1} &\leq (2^{m+1} \cdot \tau) + \alpha \cdot (2^{m+1} \cdot (m+1)^3 \cdot \tau)^{1.5} \label{eq:group-bound}
\end{align}
where $\alpha > 0$ is a sufficiently large universal constant specified below. This implies the claimed bound on $\tau_{\log s}$.

To prove the bound in Inequality \eqref{eq:group-bound}, consider $\tau_m$ in light of the recurrence relation of Equation \eqref{eq:group-recursive}. We bound $\tau_m$ by proving a bound of the following form:
\begin{align}
\tau_m &\leq 2^m \cdot \tau + \sum_{i=1}^{m} c_{m,i} \cdot \tau^{1.5^i} \label{eq:tau-c} \\
&= c_{m,0} \cdot \tau + \sum_{i=1}^{m} c_{m,i} \cdot \tau^{1.5^i} \nonumber \\
&= \sum_{i=0}^{m} c_{m,i} \cdot \tau^{1.5^i} \nonumber
\end{align}
Where the term $c_{m,0}$ is (by definition) equal to $2^m$, and we bound the terms $\{c_{m,i}\}_{m \in [1,\log s], i \in [1,m]}$ as follows. For $m \in [1,\log s],i \in [1,m]$, we show that:
\begin{align}
c_{m,i} &\leq 2 \cdot (2^m)^{1.5^i} \cdot 34^{2\cdot(1.5^{i+1} - 1.5)} \cdot {m^{2\cdot(1.5^{i+1} - 1.5)}} \label{eq:c-bound}
\end{align}
We conclude that:
\begin{align}
c_{m,i} &< (2^m)^{1.5^i} \cdot 34^{2\cdot(1.5^{i+1})} \cdot {m^{2\cdot(1.5^{i+1})}} \nonumber \\
&= 2^{m \cdot 1.5^i} \cdot 34^{3\cdot1.5^{i}} \cdot {m^{3\cdot1.5^{i}}} \nonumber \\
&= (2^m \cdot 34^3 \cdot m^3)^{1.5^i} \nonumber
\end{align}
These bounds are shown below. Plugging the bounds into equation \eqref{eq:tau-c} we get:
\begin{align}
\tau_{m+1} &\leq  ( 2^{m+1} \cdot \tau) + \sum_{i=1}^{m+1} (2^{m+1} \cdot 34^3 \cdot (m+1)^3)^{1.5^i} \cdot \tau^{1.5^i} \nonumber \\
&= (2^{m+1} \cdot \tau) + \sum_{i=1}^{m+1} (2^{m+1} \cdot 34^3 \cdot (m+1)^3 \cdot \tau)^{1.5^i} \nonumber \\
&\leq (2^{m+1} \cdot \tau) + 2(2^{m+1} \cdot 34^3 \cdot (m+1)^3 \cdot \tau)^{1.5} \nonumber \\
&= (2^{m+1} \cdot \tau) + \alpha \cdot (2^{m+1} \cdot (m+1)^3 \cdot \tau)^{1.5} \nonumber
\end{align}
where the next-to-last inequality assumes that $(\tau \cdot s \cdot \log^3 s \cdot 34^3) \leq 1/2$, and the last inequality holds for a sufficiently large universal constant $\alpha = 2 \cdot 34^{4.5}$. This proves the bound claimed in Equation \eqref{eq:group-bound}.

We prove Inequalities \eqref{eq:tau-c} and \eqref{eq:c-bound} by induction on $m$. The base case is for $m=0$, where $c_{0,0}=1$ and for $i \geq 1$ we have $c_{0,i}=0$. Assuming the bounds are true for $m$, and using Inequality \eqref{eq:group-recursive} (derived from Lemma \ref{lemma:group-standard}), we have:
\begin{align}
\tau_{m+1} &\leq 2 \tau_{m} + 34\tau_m^{1.5} \nonumber \\
&\leq 2 \left( \sum_{i=0}^m c_{m,i} \cdot \tau^{1.5^i} \right) + 34 \cdot \left( \sum_{i=0}^m c_{m,i} \cdot \tau^{1.5^i} \right)^{1.5} \nonumber \\
&\leq  2 \left( \sum_{i=0}^m c_{m,i} \cdot \tau^{1.5^i} \right) + 34 \cdot \sqrt{m+1} \cdot \left( \sum_{i=0}^m c_{m,i}^{1.5} \cdot \tau^{1.5^{i+1}} \right) \label{eq:group-std-jensen} \\
& = 2 c_{m,0} \cdot \tau + \sum_{i=1}^{m} \left( 2 c_{m,i} + (34 \sqrt{m+1} \cdot c_{m,i-1}^{1.5}) \right) \cdot \tau^{1.5} \label{eq:group-std-terms}
\end{align}
where Inequality \eqref{eq:group-std-jensen} uses the fact that for positive terms $a_i$ we have $(\sum_{i=0}^{m} a_i)^{1.5} \leq \sqrt{m+1} \cdot ( \sum_{i=0}^{m} a_i^{1.5})$, which follows from Jensen's Inequality (since $x^{1.5}$ is a convex function).

Using Inequality \eqref{eq:group-std-terms}, we can define the terms $c_{m+1,i}$ as follows:
\begin{align*}
c_{m+1,0} &= 2 c_{m,0} = 2^{m+1} \cdot \tau \\
\forall i \in [1,m], \mbox { }  c_{m+1,i} &= 2c_{m,i} + 34 \cdot \sqrt{m+1} \cdot c_{m,i-1}^{1.5} \nonumber
\end{align*}
It remains to prove Inequality \eqref{eq:c-bound} for $i \in [1,m]$. The proof is by induction. The base case for $m=0$ follows by definition. Assume for a fixed value $m$, the inequality holds $\forall i \in [1,m]$. I.e.:
\begin{align*}
c_{m,i} &\leq (2^m)^{1.5^i} \cdot 34^{2\cdot(1.5^{i+1} - 1.5)} \cdot {m^{2\cdot(1.5^{i+1} - 1.5)}}
\end{align*}
Then we get that $\forall i \in [1,m+1]$:
\begin{align*}
c_{m+1,i} &= 2c_{m,i} + 34 \cdot \sqrt{m+1} \cdot c_{m,i-1}^{1.5} \\
&\leq 2\left( (2^m)^{1.5^i} \cdot 34^{2\cdot(1.5^{i+1} - 1.5)} \cdot {m^{2\cdot(1.5^{i+1} - 1.5)}} \right) +  34 \sqrt{m+1} \cdot \left( (2^m)^{1.5^{i-1}} \cdot 34^{2\cdot(1.5^{i} - 1.5)} \cdot {m^{2\cdot(1.5^{i} - 1.5)}} \right)^{1.5} \\
&= 2\left( 2^{m \cdot 1.5^i} \cdot 34^{2\cdot(1.5^{i+1} - 1.5)} \cdot {m^{2\cdot(1.5^{i+1} - 1.5)}} \right) + 34 \sqrt{m+1} \cdot \left( 2^{m \cdot 1.5^i} \cdot 34^{2\cdot(1.5^{i+1} - 1.5^2)} \cdot {m^{2\cdot(1.5^{i+1} - 1.5^2)}} \right) \\
&< 2\left( 2^{m \cdot 1.5^i} \cdot 34^{2\cdot(1.5^{i+1} - 1.5)} \cdot {m^{2\cdot(1.5^{i+1} - 1.5)}} \right) + \left( 2^{m \cdot 1.5^i} \cdot 34^{2\cdot(1.5^{i+1} - 1.5^2) + 1} \cdot {(m+1)^{2\cdot(1.5^{i+1} - 1.5^2)+1/2}} \right) \\
&= 2\left( 2^{m \cdot 1.5^i} \cdot 34^{2\cdot(1.5^{i+1} - 1.5)} \cdot {m^{2\cdot(1.5^{i+1} - 1.5)}} \right) + \left( 2^{m \cdot 1.5^i} \cdot 34^{2\cdot(1.5^{i+1} - 1.5) - 0.5} \cdot {(m+1)^{2\cdot(1.5^{i+1} - 1.5) - 1}} \right) \\
&= 2\left( 2^{m \cdot 1.5^i} \cdot 34^{2\cdot(1.5^{i+1} - 1.5)} \cdot {m^{2\cdot(1.5^{i+1} - 1.5)}} \right) + \frac{ \left( 2^{m \cdot 1.5^i} \cdot 34^{2\cdot(1.5^{i+1} - 1.5)} \cdot {(m+1)^{2\cdot(1.5^{i+1} - 1.5)}} \right)} { 34^{0.5} \cdot (m+1) } \\
&< 2.5 \cdot 2^{m \cdot 1.5^i} \cdot 34^{2\cdot(1.5^{i+1} - 1)} \cdot {(m+1)^{2\cdot(1.5^{i+1} - 1/2)}} \\
&= 2.5 \cdot 2^{((m+1) \cdot 1.5^i) -1.5^i} \cdot 34^{2\cdot(1.5^{i+1} - 1.5)} \cdot {(m+1)^{2\cdot(1.5^{i+1} - 1.5)}} \\
&= \frac{2.5}{2^{1.5^i}} \cdot 2^{(m+1) \cdot 1.5^i } \cdot 34^{2\cdot(1.5^{i+1} - 1.5)} \cdot {(m+1)^{2\cdot(1.5^{i+1} - 1.5)}} \\
&< 2^{(m+1) \cdot 1.5^i } \cdot 34^{2\cdot(1.5^{i+1} - 1.5)} \cdot {(m+1)^{2\cdot(1.5^{i+1} - 1.5)}}
\end{align*}

\paragraph{Bounding $\mu_{m+1}$.}
We prove that:
\begin{align}
\mu_{m+1} &\leq \frac{(2^{m+1} \cdot \tau)^2}{2} +  \alpha \cdot (2^{m+1} \cdot \tau)^{2.5} \cdot m^{4.5} \label{eq:group-exp-bound}.
\end{align}
(where $\alpha$ is the universal constant specified above). The proof will be by induction over $m$. For the base case $m=0$, we know that $\mu \leq \tau^2/2$ by the lemma conditions. For the induction step, by Lemma \ref{lemma:group-expectation} and using $\mu_m \leq \tau_m^2/2$:
\begin{align}
\mu_{m+1} &\leq 2\mu_{m} + \tau_{m}^2 + 3.5\tau_m^{3} + 1.5\tau_m^4. \label{eq:group-exp-recursive}
\end{align}
Using the induction hypothesis and the bound on the subgaussian standard $\tau_m$ shown above (Inequality \eqref{eq:group-bound}), we conclude (so long as $\tau_m$ is a sufficiently small constant) that:
\begin{align*}
\mu_{m+1} &\leq 2\mu_{m} + \left( \tau_{m}^2 + 3.5\tau_m^{3} + 1.5\tau_m^4 \right) \\
&\leq 2\mu_{m} + \left( (2^{m} \cdot \tau)^2 + 3\alpha \cdot ((2^{m} \cdot \tau)^{2.5} \cdot m^{4.5}) \right)\\
&\leq \left((2^{m} \cdot \tau)^2 + 2\alpha \cdot (2^{m} \cdot \tau)^{2.5} \cdot m^{4.5} \right)+ \left( (2^{m} \cdot \tau)^2 + 3\alpha \cdot ((2^{m} \cdot \tau)^{2.5} \cdot m^{4.5}) \right) \\
&= 2(2^{m} \cdot \tau)^2 + 5\alpha \cdot ((2^{m} \cdot \tau)^{2.5} \cdot m^{4.5}) \\
& < \frac{(2^{m+1} \cdot \tau)^2}{2} + \alpha \cdot (2^{m+1} \cdot \tau)^{2.5} \cdot m^{4.5}.
\end{align*}
This implies the claimed bound on $\mu_{\log s}$ (since $(s \cdot \polylog s \cdot \tau )< 1$).

\paragraph{Relationship between $\mu_{m+1}$ and $\tau_{m+1}$.} Finally, we show that the above bounds maintain that $\mu_{m+1} \leq \tau_{m+1}^2/2$. To see this:
\begin{align*}
\tau_{m+1}^2 &=  \left((2^{m+1} \cdot \tau) + \alpha \cdot (2^{m+1} \cdot (m+1)^3 \cdot \tau)^{1.5} \right) \\
&\geq (2^{m+1} \cdot \tau)^2 + 2\alpha \cdot (2^{m+1} \cdot \tau) \cdot (2^{m+1} \cdot (m+1)^3 \cdot \tau)^{1.5} \\
&= (2^{m+1} \cdot \tau)^2 +  2\alpha \cdot (2^{m+1} \cdot \tau)^{2.5} \cdot m^{4.5} \\
&\geq 2\mu_{m+1}
\end{align*}

\end{proof}

\subsection{Group Privacy: Bounding the Expected Privacy Loss}
\label{subsec:group-expectation}

In this section we bound the {\em expected} privacy loss for groups, using the following Lemma:
\begin{lemma} \label{lemma:group-expectation}
Let $D,D',D''$ be distributions over domain $\X$, such that
$\DsubG(D||D'),  \DsubG(D'||D) \preceq (\mu_1,\sub_1)$ and that $\DsubG(D'||D''),  \DsubG(D''||D') \preceq (\mu_2,\sub_2)$. Suppose moreover that $\sub_1 \leq 1/3$. Then it is also the case that:
$$\DKL(D||D'') \leq \mu_1 + \mu_2 + \tau_1 \cdot \tau_2 + \left( (2\tau_1^2 \cdot \tau_2) + ((\tau_1 + 3\tau_1^2) \cdot \mu_2) \right)$$
\end{lemma}

\begin{proof}
For $x \in \X$, we define $S(x)$ to be the {\em centered} value $S(x) = \ln \frac{D'[x]}{D[x]} - \DKL(D'||D)$, and $S''(x)$ to be the centered value
$S''(x) = \ln \frac{D'[x]}{D''[x]} - \DKL(D'||D'')$. When $x$ is drawn by $D'[x]$, both $S[x]$ and $S''[x]$ are (centered) subgaussian random variables. We use $\Var(S),\Var(S'')$ to denote the variances of these random variables (which are bounded by $\sub_1^2,\sub_2^2$ respectively). We have that:
\begin{align}
\DKL(D||D'') &= \sum_{x \in \X} \left( D[x] \cdot \ln \frac{D[x]}{D''[x]} \right)  \nonumber \\
&= \sum_{x \in \X} \left( D[x] \cdot \big( \ln \frac{D[x]}{D'[x]} +  \ln \frac{D'[x]}{D''[x]} \big) \right) \nonumber \\
&=  \sum_{x \in \X} \left( D[x] \cdot \ln \frac{D[x]}{D'[x]} \right)  + \sum_{x \in \X} \left( D[x] \cdot \ln \frac{D'[x]}{D''[x]} \right) \label{eq:group-Exp-1}
\end{align}
The first of these summands is $\DKL(D||D')=\mu_1$. We bound the second summand:
\begin{align}
\sum_{x \in \X} \left( D[x] \cdot \ln \frac{D'[x]}{D''[x]} \right) &= \sum_{x \in \X} \left( D'[x] \cdot \frac{D[x]}{D'[x]} \cdot \ln \frac{D'[x]}{D''[x]} \right) \nonumber \\
&= \sum_{x \in \X} \left( D'[x] \cdot e^{-ln \frac{D'[x]}{D[x]}} \cdot \ln \frac{D'[x]}{D''[x]} \right) \nonumber \\
&= \sum_{x \in \X} \left( D'[x] \cdot e^{- (ln \frac{D'[x]}{D[x]} - \DKL(D'||D) + \DKL(D'||D))} \cdot \ln \frac{D'[x]}{D''[x]} \right) \nonumber \\
&= e^{-\DKL(D'||D)} \cdot \sum_{x \in \X} \left( D'[x] \cdot e^{-S(x)} \cdot \ln \frac{D'[x]}{D''[x]} \right) \nonumber \\
&\leq \sum_{x \in \X} \left( D'[x] \cdot e^{-S(x)} \cdot \ln \frac{D'[x]}{D''[x]} \right) \nonumber \\
&= \Ex_{x \sim D'} \left[ \ln \frac{D'[x]}{D''[x]} \cdot e^{-S(x)} \right] \nonumber
\end{align}
where the next-to-last inequality is by non-negativity of KL-divergence. Recall that $S(x)$ denotes the {\em centered} log-ratio of probabilities by $D'$ and by $D$ (and is $\tau_1$-Subgaussian). By Lemma \ref{lemma:subgauss-product} we conclude that:
\begin{align}
\sum_{x \in \X} \left( D[x] \cdot \ln \frac{D'[x]}{D''[x]} \right) &\leq \Ex_{x \sim D'} \left[ \ln \frac{D'[x]}{D''[x]} \cdot e^{-S(x)} \right] \nonumber \nonumber \\
&\leq \DKL(D'||D'') + \left( \sqrt{\Var(S'')} + \DKL(D'||D'') \right) \cdot (\tau_1 + 3\tau_1^2) \nonumber \\
&\leq \mu_2 + \left( \sqrt{\tau_2^2} + \mu_2 \right) \cdot (\tau_1 + 3\tau_1^2) \nonumber \\
&= \mu_2 + \tau_1 \cdot \tau_2 + \left( (3\tau_1^2 \cdot \tau_2) + ((\tau_1 + 3\tau_1^2) \cdot \mu_2) \right) \label{eq:group-Exp-2}
\end{align}

Putting together Equations \eqref{eq:group-Exp-1} and \eqref{eq:group-Exp-2} we get that:
\begin{align}
\DKL(D||D'') &\leq \mu_1 + \mu_2 + \tau_1 \cdot \tau_2 + \left( (3\tau_1^2 \cdot \tau_2) + ((\tau_1 + 3\tau_1^2) \cdot \mu_2) \right) \nonumber
\end{align}
\end{proof}

\subsection{Group Privacy: Bounding the Subgaussian Standard}
\label{subsec:group-standard}

\begin{lemma} \label{lemma:group-standard}
Let $D,D',D''$ be distributions over domain $\X$, such that
$\DsubG(D||D'),  \DsubG(D'||D) \preceq (\mu_1,\sub_1)$ and that $\DsubG(D'||D''),  \DsubG(D''||D') \preceq (\mu_2,\sub_2)$. Suppose moreover that $\sub_1,\sub_2 \leq \sub \leq 1/4$ and that $\mu_1,\mu_2 \leq \sub^2/2$. Then for any real $\lambda$:
\begin{align}
\Ex_{x \sim D} \left[ e^{\lambda \cdot (\ln \frac{D[x]}{D''[x]} - \DKL(D||D''))} \right] &\leq e^{\frac{\lambda^2}{2} \cdot (2\tau + 34\tau^{1.5})^2 } \label{eq:group-std-2}
\end{align}
i.e. the (centered) privacy-loss random variable from $D$ to $D''$ is subgaussian, and its standard is bounded by $2\tau + O(\tau^{1.5})$.
\end{lemma}

\begin{proof}
We assume without loss of generality that:
\begin{align}
& \lambda \cdot \DKL(D'||D) \geq \lambda \cdot \DKL(D'||D'') \nonumber \\
& \Rightarrow (\lambda+1) \cdot \DKL(D'||D) \geq \lambda \cdot \DKL(D'||D'')  \label{eq:group-std-1}
\end{align}
(otherwise we flip the roles of $D$ and $D''$). As in the proof of Lemma \ref{lemma:group-expectation}, for $x \in \X$, we define $S$ to be the {\em centered} value $S(x) = \ln \frac{D'[x]}{D[x]} - \DKL(D'||D)$, and $S''(x)$ to be the centered value
$S''(x) = \ln \frac{D'[x]}{D''[x]} - \DKL(D'||D'')$. Recall that when $x$ is drawn by $D'[x]$, both $S[x]$ and $S''[x]$ are (centered) subgaussian random variables.

We want to show that Inequality \eqref{eq:group-std-2} holds for {\em any} real $\lambda$. We proceed with a case analysis for the value of $\lambda$ as a function of $\tau$.

\paragraph{Case $I$: $|\lambda| \geq \frac{1}{8\sqrt{\tau}}$.} Observe that:
\begin{align}
\Ex_{x \sim D} \left[ e^{\lambda \cdot \ln \frac{D[x]}{D''[x]}} \right] &= \sum_{x \in \X} \left( D[x] \cdot e^{\lambda \cdot \ln \frac{D[x]}{D''[x]}} \right) \nonumber \\
&= \sum_{x \in \X} \left( D'[x] \cdot e^{\ln \frac{D[x]}{D'[x]}} \cdot e^{\lambda \cdot (\ln \frac{D[x]}{D'[x]} + \ln \frac{D'[x]}{D''[x]})} \right) \nonumber \\
&= \sum_{x \in \X} \left( D'[x] \cdot e^{-(\lambda+1) \cdot \ln \frac{D'[x]}{D[x]}} \cdot e^{\lambda \cdot \ln \frac{D'[x]}{D''[x]}} \right) \nonumber \\
&= \sum_{x \in \X} \left( D'[x] \cdot e^{-(\lambda+1) \cdot (S(x) + \DKL(D'||D))} \cdot e^{\lambda \cdot (S''(x) + \DKL(D'||D''))} \right) \nonumber \\
&= \left( e^{(-(\lambda+1) \cdot \DKL(D'||D)) + (\lambda \cdot \DKL(D'||D'')) }\right) \cdot \sum_{x \in \X} \left( D'[x] \cdot e^{-(\lambda+1) \cdot S(x)} \cdot e^{\lambda \cdot S''(x)} \right) \nonumber \\
&\leq \sum_{x \in \X} \left( D'[x] \cdot e^{-(\lambda+1) \cdot S(x)} \cdot e^{\lambda \cdot S''(x)} \right) \nonumber
\end{align}
where the last inequality is by Equation \eqref{eq:group-std-1}. Recall that $S(x)$ and $S''(x)$ are both subgaussian (when $x$ is drawn by $D'[x]$), with standards $\tau_1,\tau_2$. Applying the Cauchy-Schwartz inequality to the above, we conclude:
\begin{align}
\Ex_{x \sim D} \left[ e^{\lambda \cdot \ln \frac{D[x]}{D''[x]}} \right] &\leq \sqrt{\sum_{x \in \X}  D'[x] \cdot e^{-2(\lambda+1) \cdot S(x)}} \cdot \sqrt{\sum_{x \in \X} D'[x] \cdot e^{-2\lambda \cdot S''(x)}} \nonumber \\
&\leq \sqrt{e^{(2\lambda^2 + 4\lambda + 2) \cdot \tau_1^2}} \cdot \sqrt{e^{2\lambda^2 \cdot \tau_2^2}} \nonumber \\
&= e^{(\lambda+1)^2 \cdot \tau_1^2 + \lambda^2 \cdot \tau_2^2} \label{eq:group-std-9}
\end{align}
Since $|\lambda| \geq \frac{1}{8\sqrt{\tau}}$, we have that:
\begin{align*}
(\lambda+1)^2 &= \lambda^2 \cdot (1 + \frac{1}{\lambda})^2 \\
&\leq \lambda^2 \cdot (1 + 8\sqrt{\tau})^2 \\
&= \lambda^2 \cdot (1 + 16\sqrt{\tau} + 64\tau) \\
&\leq \lambda^2 \cdot (1 + 48\sqrt{\tau})
\end{align*}
(recall also that $\tau \leq 1/4$, so $\tau \leq \sqrt{\tau}/2$). Plugging this into Equation \eqref{eq:group-std-9}, we get that:
\begin{align}
\Ex_{x \sim D} \left[ e^{\lambda \cdot \ln \frac{D[x]}{D''[x]}} \right] &\leq  e^{(\lambda+1)^2 \cdot \tau_1^2 + \lambda^2 \cdot \tau_2^2} \nonumber \\
&\leq e^{(\lambda^2 \cdot \tau_1^2 \cdot (1 + 48\sqrt{\tau})) + (\lambda^2 \cdot \tau_2^2)} \nonumber \\
&\leq e^{\lambda^2 \cdot (\tau_1^2 + \tau_2^2 + 48\tau^{2.5})} \nonumber \\
&\leq e^{\frac{\lambda^2}{2} \cdot (2\tau + 12\tau^{1.5})^2}. \nonumber
\end{align}
We conclude that for $\lambda \geq \frac{1}{8\sqrt{\tau}}$, the ``standard'' is bounded by $(2\tau + 12\tau^{1.5})$.

\paragraph{Case $II$: $|\lambda| < \frac{1}{8\sqrt{\tau}}$}
Taking a Taylor expansion, we get that:
\begin{align}
& \Ex_{x \sim D} \left[ e^{\lambda \cdot \ln \frac{D[x]}{D''[x]}} \right] = 1 + \lambda \cdot \DKL(D||D'') + \frac{\lambda^2}{2} \cdot \Ex_{x \sim D} \left[ \ln^2 \frac{D[x]}{D''[x]} \right] + \sum_{k=3}^{\infty} \frac{\lambda^k}{k!} \cdot \Ex_{x \sim D} \left[ \ln^k \frac{D[x]}{D''[x]} \right] \label{eq:group-std-3}
\end{align}
(where we observe that the linear summand in the Taylor expansion is the expected log ratio or ``privacy loss'' from $D$ to $D''$). In the following two claims, we bound the higher moments in the Taylor expansion:

\begin{claim} \label{claim:term-2}
\begin{align}
\Ex_{x \sim D} \left[ \ln^2 \frac{D[x]}{D''[x]} \right] &\leq 2(\tau_1^2 + \tau_2^2) + 2\mu_1^2 + 2\mu_2^2 + 50\tau_1 \cdot \tau_2^2 \nonumber \\
&\leq 4\tau^2 + 51 \tau^3 \nonumber \\
&\leq 4\tau^2 + 26 \tau^{2.5}
\end{align}
\end{claim}

\begin{claim} \label{claim:term-3}
\begin{align}
\sum_{k=3}^{\infty} \frac{\lambda^k}{k!} \cdot \Ex_{x \sim D} \left[ \ln^k \frac{D[x]}{D''[x]} \right] &\leq (33 \cdot \tau_1^{2.5} \cdot \lambda^2) + (32 \cdot \tau_2^{2.5} \cdot \lambda^2) \nonumber \\
&\leq 55 \cdot \tau^{2.5} \cdot \lambda^2 \nonumber
\end{align}
\end{claim}

The proofs of Claims \ref{claim:term-2} and \ref{claim:term-3} follow below. Before presenting these proofs, we complete the proof of the bound for case $II$. Plugging the bounds from the Claims into Inequality \eqref{eq:group-std-3} we get:
\begin{align}
\Ex_{x \sim D} \left[ e^{\lambda \cdot \ln \frac{D[x]}{D''[x]}} \right] &= 1 + \lambda \cdot \DKL(D||D'') + \frac{\lambda^2 \cdot (4\tau^2 + 26 \tau^{2.5})}{2} + 55 \tau^{2.5} \cdot \lambda^2 \nonumber \\
 &= 1 + \lambda \cdot \DKL(D||D'') + \frac{4\lambda^2 \cdot \tau^2}{2} + 68 \tau^{2.5} \cdot \lambda^2 \nonumber \\
&\leq e^{\lambda \cdot \DKL(D||D'') + \frac{\lambda^2}{2} \cdot ((2\tau)^2 + 136 \tau^{2.5})} \nonumber \\
&\leq e^{\lambda \cdot \DKL(D||D'') + \frac{\lambda^2}{2} \cdot (2\tau + 34\tau^{1.5})^2 } \nonumber
\end{align}
Thus, for the {\em centered} privacy-loss random variable we have:
\begin{align}
\Ex_{x \sim D} \left[ e^{\lambda \cdot (\ln \frac{D[x]}{D''[x]} - \DKL(D||D''))} \right] &\leq e^{\frac{\lambda^2}{2} \cdot (2\tau + 34\tau^{1.5})^2 } \nonumber
\end{align}
We conclude that for $\lambda < \frac{1}{8\sqrt{\tau}}$, the ``standard'' is bounded by $(2\tau + 34\tau^{1.5})$.

Before proceeding to prove the claims, we state and prove the following useful fact:
\begin{fact} \label{fact:abpowers}
For a real value $k \geq 1$ and for any two reals $a,b$:
$$(a+b)^k \leq 2^{k-1} \cdot (a^k + b^k)$$
\end{fact}

\begin{proof}
Since the function $x^k$ is convex, by Jensen's inequality we have:
\begin{align*}
(a+b)^k = \left( \frac{2a}{2}  + \frac{2b}{2} \right)^k \leq \frac{(2a)^k + (2b)^k}{2} = 2^{k-1} \cdot (a^k + b^k)
\end{align*}
\end{proof}

\begin{proof}[Proof of Claim \ref{claim:term-2}]
We observe that, since $(a+b)^2 \leq 2a^2 + 2b^2$ (for all $a,b$, see Fact \ref{fact:abpowers}):
\begin{align}
\Ex_{x \sim D} \left[ \ln^2 \frac{D[x]}{D''[x]} \right] &= \Ex_{x \sim D} \left[ (\ln \frac{D[x]}{D'[x]} + \ln \frac{D'[x]}{D''[x]})^2 \right] \nonumber \\
&\leq 2\Ex_{x \sim D} \left[ \ln^2 \frac{D[x]}{D'[x]} \right] + 2\Ex_{x \sim D} \left[ \ln^2 \frac{D'[x]}{D''[x]} \right] \nonumber \\
& \leq 2(\tau_1^2+\mu_1^2) + 2\Ex_{x \sim D} \left[ \ln^2 \frac{D'[x]}{D''[x]} \right]\label{eq:group-std-4}
\end{align}
where the last inequality uses Fact \ref{fact:subgauss-variance} (and the fact that $\Ex[X^2] = \Var(X) + (E[X])^2$ for any RV $X$). To bound the second summand in Inequality \eqref{eq:group-std-4}, we use Lemma \ref{lemma:subgauss-product} and conclude that:
\begin{align}
\Ex_{x \sim D} \left[ \ln^2 \frac{D'[x]}{D''[x]} \right] &= \Ex_{x \sim D'} \left[  \ln^2 \frac{D'[x]}{D''[x]} \cdot e^{-\ln \frac{D'[x]}{D[x]}} \right] \nonumber \\
&\leq  \Ex_{x \sim D'} \left[  \ln^2 \frac{D'[x]}{D''[x]} \right] + \sqrt{\Ex_{x \sim D'} \left[  \ln^4 \frac{D'[x]}{D''[x]} \right]} \cdot (\tau_1 + 3\tau_1^2) \nonumber \\
&=\Ex_{x \sim D'} \left[  \ln^2 \frac{D'[x]}{D''[x]} \right] + \sqrt{\Ex_{x \sim D'} \left[ (S''(x) + \DKL(D'||D''))^4 \right]} \cdot (\tau_1 + 3\tau_1^2) \nonumber \\
&\leq \Ex_{x \sim D'} \left[  \ln^2 \frac{D'[x]}{D''[x]} \right] + \sqrt{8 \cdot(\Ex_{x \sim D'} \left[ S''(x)^4 \right] + \DKL(D'||D'')^4 )} \cdot (\tau_1 + 3\tau_1^2) \nonumber
\end{align}
where the last inequality uses Fact \ref{fact:abpowers} (with $k=4$). By the bound on the 4-th moment of the subgaussian $S''(x)$ (Fact \ref{fact:subgauss-moments}), we conclude that:
\begin{align}
\Ex_{x \sim D} \left[ \ln^2 \frac{D'[x]}{D''[x]} \right] &\leq (\tau_2^2+\mu_2^2) + \sqrt{8 \cdot (16\tau_2^4 + \mu_2^4)} \cdot (\tau_1 + 3\tau_1^2) \nonumber \\
&\leq \nonumber (\tau_2^2+\mu_2^2) + \sqrt{8 \cdot (16\tau_2^4 + \mu_2^4)} \cdot 2\tau_1 \\
&\leq (\tau_2^2+\mu_2^2) + 24\tau_1 \cdot \tau_2^2 + 6 \tau_1 \cdot \mu_2^2 \nonumber \\
&\leq (\tau_2^2+\mu_2^2) + 25\tau_1 \cdot \tau_2^2 \label{eq:group-std-5}
\end{align}
Claim \ref{claim:term-2} follows from inequalities \eqref{eq:group-std-4} and \eqref{eq:group-std-5}:
\begin{align}
\Ex_{x \sim D} \left[ \ln^2 \frac{D[x]}{D''[x]} \right] &\leq 2(\tau_1^2+\mu_1^2)+ 2(\tau_2^2+\mu_2^2 + 25\tau_1 \cdot \tau_2^2) \nonumber \\
&\leq 2(\tau_1^2 + \tau_2^2) + 2\mu_1^2 + 2\mu_2^2 + 50\tau_1 \cdot \tau_2^2 \nonumber
\end{align}
\end{proof}

\begin{proof}[Proof of Claim \ref{claim:term-3}]
Observe that, by Fact \ref{fact:abpowers}:
\begin{align}
\Ex_{x \sim D} \left[ \ln^k \frac{D[x]}{D''[x]} \right] &= \Ex_{x \sim D}  \left[ (\ln \frac{D[x]}{D'[x]} + \ln \frac{D'[x]}{D''[x]})^k \right] \nonumber \\
&\leq \Ex_{x \sim D} \left[ 2^{k-1} \cdot(\ln^k \frac{D[x]}{D'[x]} + \ln^k \frac{D'[x]}{D''[x]}) \right] \nonumber\\
&= 2^{k-1} \cdot \left( \Ex_{x \sim D} \left[  \ln^k \frac{D[x]}{D'[x]} \right]  + \Ex_{x \sim D} \left[ \ln^k \frac{D'[x]}{D''[x]} \right] \right) \label{eq:group-std-6}
\end{align}
By Fact \ref{fact:subgauss-moments}, and using also Fact \ref{fact:abpowers}, the first term is bounded by:
\begin{align}
\Ex_{x \sim D} \left[  \ln^k \frac{D[x]}{D'[x]} \right] &\leq 2^{k-1} \cdot \left( \Ex_{x \sim D} \left[ S(x)^k + \DKL(D||D')^k \right]\right) \nonumber \\
&\leq (2^{k-1} \cdot 2^{\lceil k/2 \rceil +1} \cdot (\lceil k/2 \rceil)! \cdot \tau_1^k) + (2\mu_1)^k/2 \nonumber \\
&= (2\tau_1)^k \cdot  2^{\lceil k/2 \rceil} \cdot (\lceil k/2 \rceil)! + (2\mu_1)^k/2 \nonumber \\
&\leq  \left( (2\tau_1)^k \cdot 2^{k/2+1} \cdot \sqrt{k!}\right) + (2\mu_1)^k/2 \label{eq:group-std-7}
\end{align}
We bound the second term using Lemma \ref{lemma:subgauss-product} and Fact \ref{fact:abpowers}:
\begin{align}
\Ex_{x \sim D} \left[ \ln^k \frac{D'[x]}{D''[x]} \right] &= \Ex_{x \sim D'} \left[ \ln^k \frac{D'[x]}{D''[x]} \cdot e^{-\ln \frac{D[x]}{D'[x]}} \right] \nonumber \\
&\leq \Ex_{x \sim D'} \left[ \ln^k \frac{D'[x]}{D''[x]} \right] + \sqrt{\Ex_{x \sim D'} \left[ \ln^{2k} \frac{D'[x]}{D''[x]} \right]} \cdot (\tau_1 + 3\tau_1^2) \nonumber \\
&\leq 2^{k-1} \cdot \Ex_{x \sim D'} \left[ S''(x)^k + \DKL(D'||D'')^k \right] +  \sqrt{\Ex_{x \sim D'} \left[ \ln^{2k} \frac{D'[x]}{D''[x]} \right]} \cdot (\tau_1 + 3\tau_1^2) \nonumber \\
&\leq 2^{k-1} \left( \Ex_{x \sim D'} \left[ S''(x)^k \right]  + \mu_2^k \right)  +  \sqrt{\Ex_{x \sim D'} \left[ \ln^{2k} \frac{D'[x]}{D''[x]} \right]} \cdot (\tau_1 + 3\tau_1^2) \nonumber \\
&\leq 2^{k-1} \left( \Ex_{x \sim D'} \left[ S''(x)^k \right]  + \mu_2^k \right) + \sqrt{  2^{2k-1} \cdot \Ex_{x \sim D'} \left[ S''(x)^{2k} + \DKL(D'||D'')^{2k} \right] } \cdot (\tau_1 + 3\tau_1^2) \nonumber \\
&\leq  2^{k-1} \cdot  \left( \Ex_{x \sim D'} \left[ S''(x)^k \right]  + \mu_2^k \right) + \left( \sqrt{  2^{2k-1} \cdot \Ex_{x \sim D'} [ S''(x)^{2k} ] }+ \sqrt{2^{2k-1} \cdot \mu_2^{2k}} \right) \cdot (\tau_1 + 3\tau_1^2) \nonumber \\
&\leq 2^{k-1} \cdot \Ex_{x \sim D'} \left[ S''(x)^k \right] + \sqrt{  2^{2k-1} \cdot \Ex_{x \sim D'} \left[ S''(x)^{2k}\right] } \cdot (\tau_1 + 3\tau_1^2) + (2\mu_2)^k, \nonumber
\end{align}
where the last inequality follows because $\tau_1 + 3\tau_1^2 \leq 7/16 < 1/2$.
By Fact \ref{fact:subgauss-moments}, we conclude from the above that:
\begin{align}
\Ex_{x \sim D} \left[ \ln^k \frac{D'[x]}{D''[x]} \right] &\leq \left((2\tau_2)^k \cdot  2^{\lceil k/2 \rceil} \cdot (\lceil k/2 \rceil)! \right) + \left( \sqrt{2^{2k-1} \cdot 2^{k+1} \cdot (k!) \cdot \tau_2^{2k}}  \cdot (\tau_1 + 3\tau_1^2)  \right) + (2\mu_2)^k \nonumber \\
&= \left((2\tau_2)^k \cdot  2^{\lceil k/2 \rceil} \cdot (\lceil k/2 \rceil)! \right) + \left((2\tau_2)^k \cdot 2^{k/2} \cdot \sqrt{k!} \cdot (\tau_1 + 3\tau_1^2)  \right) + (2\mu_2)^k \nonumber \\
&\leq \left((2\tau_2)^k  \cdot  2^{k/2+1} \cdot \sqrt{k!} \right) + (2\mu_2)^k \label{eq:group-std-8}
\end{align}
Plugging the bounds from Equations \eqref{eq:group-std-7} and \eqref{eq:group-std-8} into Equation \eqref{eq:group-std-6}, we conclude that:
\begin{align}
\Ex_{x \sim D} \left[ \ln^k \frac{D[x]}{D''[x]} \right] &\leq \left((4\tau_1)^k + (4\tau_2)^k \right) \cdot \left( 2^{k/2} \cdot \sqrt{k!} \right) + (2\mu_2)^k + (2\mu_2)^k \nonumber \\
&\leq \left((4\tau_1)^k + (4\tau_2)^k \right) \cdot \left( 2^{k/2} \cdot \sqrt{k!} \right) + 2 \cdot (2\mu_1)^k
\end{align}
Using the fact that for $k \geq 3$ we have $\frac{2^{k/2} \cdot \sqrt{k!}}{k!} < 2$, we get that:
\begin{align}
\sum_{k=3}^{\infty} \frac{\lambda^k}{k!} \cdot \Ex_{x \sim D} \left[ \ln^k \frac{D[x]}{D''[x]} \right] &\leq \sum_{k=3}^{\infty} 2\left((4\tau_1 \cdot \lambda)^k + (4\tau_2 \cdot \lambda)^k \right)  + 2\cdot \frac{(2\mu_1 \cdot \lambda)^k}{k!} \nonumber \\
&= \left( 2\cdot (4\tau_1 \cdot \lambda)^3 \cdot \sum_{k=0}^{\infty} (4\tau_1 \cdot \lambda)^k \right) + \left( 2\cdot (4\tau_2 \cdot \lambda)^3 \cdot \sum_{k=0}^{\infty} (4\tau_1 \cdot \lambda)^k \right) + (\mu_1 \cdot \lambda)^2 \nonumber
\end{align}
where the last inequality uses the fact that for $k \geq 3$, we have $\frac{(2\mu_1 \cdot \lambda)^k}{k!} < (\mu_1 \cdot \lambda)^2 \cdot (1/4)^k$ (because we assume here that $\lambda < 1/8\sqrt{\tau}$, so $\mu_1 \cdot \lambda \leq \tau^{1.5}/16 < 1/128$). Moreover, since $4\tau_1 \cdot \lambda, 4\tau_2 \cdot \lambda \leq 1/2$, the geometric sums above converge to a value smaller than 2, and we get:
\begin{align}
\sum_{k=3}^{\infty} \frac{\lambda^k}{k!} \cdot \Ex_{x \sim D} \left[ \ln^k \frac{D[x]}{D''[x]} \right] &\leq 4^4 \cdot (\tau_1 \cdot \lambda)^3 + 4^4 \cdot (\tau_2 \cdot \lambda)^3 + (\mu_2 \cdot \lambda)^2 \nonumber
\end{align}
Finally, since $\sqrt{\tau_1}\cdot \lambda,\sqrt{\tau_2}\cdot \lambda \leq 1/8$, and $\mu_2 \leq \tau_1^2/2$, we get that:
\begin{align}
\sum_{k=3}^{\infty} \frac{\lambda^k}{k!} \cdot \Ex_{x \sim D} \left[ \ln^k \frac{D[x]}{D''[x]} \right] &\leq (32 \cdot \tau_1^{2.5} \cdot \lambda^2) + (32 \cdot \tau_2^{2.5} \cdot \lambda^2) + (\tau_1^4 \cdot \lambda^2) \nonumber \\
&\leq (33 \cdot \tau_1^{2.5} \cdot \lambda^2) + (32 \cdot \tau_2^{2.5} \cdot \lambda^2) \nonumber
\end{align}
\end{proof}
\end{proof}

\remove{
\section{Applications}

\subsection{Cohorts or First $k$ Above Threshold}

\subsection{Top $k$ Counts}
}

\bibliographystyle{alpha}

\bibliography{refs}

\end{document}